\RequirePackage{amsmath}
\documentclass[runningheads]{llncs}
\usepackage{geometry}
\usepackage{multicol}
\usepackage{amsfonts}
\usepackage{amssymb}
\usepackage{graphicx}
\usepackage{epic}
\usepackage{eepic}
\usepackage{epsfig,float}
\usepackage{verbatim}
\usepackage{pdfsync}
\usepackage{gastex}
\usepackage[lowtilde]{url}

\usepackage{xcolor}

\renewcommand{\le}{\leqslant}
\renewcommand{\ge}{\geqslant}

\newcommand{\eps}{\varepsilon}
\newcommand{\emp}{\emptyset}

\newcommand{\Sig}{\Sigma}
\newcommand{\sig}{\sigma}
\newcommand{\noin}{\noindent}

\newcommand{\bi}{\begin{itemize}}
\newcommand{\ei}{\end{itemize}}
\newcommand{\be}{\begin{enumerate}}
\newcommand{\ee}{\end{enumerate}}
\newcommand{\bd}{\begin{description}}
\newcommand{\ed}{\end{description}}
\newcommand{\bq}{\begin{quote}}
\newcommand{\eq}{\end{quote}}
\newcommand{\txt}[1]{\mbox{ #1 }}
\newcommand{\Bsf}{\mathbf{B}_{\mathrm{sf}}}

\newcommand{\Vsf}{\mathbf{T}^{\le 5}} 
\newcommand{\Wsf}{\mathbf{T}^{\ge 6}} 
\newcommand{\tid}{\mbox{{\bf 1}}}

\newcommand{\cD}{{\mathcal D}}

\newcommand{\cT}{{\mathcal T}}
\newcommand{\cW}{{\mathcal W}}

\newcommand{\qedb}{\hfill$\blacksquare$}

\title{Syntactic Complexity of Suffix-Free Languages}

\author{Janusz~Brzozowski\inst{1} \and Marek Szyku{\l}a \inst{2}}

\authorrunning{J. Brzozowski and M. Szyku{\l}a}   

\institute{David R. Cheriton School of Computer Science, University of Waterloo,\\
Waterloo, ON, Canada N2L 3G1\\
\{{\tt brzozo@uwaterloo.ca}\}
\and
Institute of Computer Science, University of Wroc{\l}aw,\\
Joliot-Curie 15, PL-50-383 Wroc{\l}aw, Poland\\
\{{\tt msz@cs.uni.wroc.pl}\}
}

\begin{document}
\maketitle

\begin{abstract}
We solve an open problem concerning syntactic complexity:
We prove that the cardinality of the syntactic semigroup of a suffix-free language with $n$ left quotients (that is, with state complexity $n$) is at most $(n-1)^{n-2}+n-2$ for $n\ge 6$.
Since this bound is known to be reachable, this settles the problem.
We also reduce the alphabet of the witness languages reaching this bound to five letters instead of $n+2$, and show that it cannot be any smaller.
Finally, we prove that the transition semigroup of a minimal deterministic automaton accepting a witness language is unique for each $n$.

\medskip
\noin{\bf Keywords:} regular language, suffix-free, syntactic complexity, transition semigroup, upper bound
\end{abstract}

\section{Introduction}

The \emph{syntactic complexity}~\cite{BrYe11} $\sig(L)$ of a regular language $L$ is the size of its syntactic semigroup~\cite{Pin97}.
This semigroup is isomorphic to the transition semigroup of the quotient automaton $\cD$ (a minimal deterministic finite automaton) accepting the language. 
The number $n$ of states of $\cD$ is the \emph{state complexity} of the language~\cite{Yu01}, and it is the same as the \emph{quotient complexity}~\cite{Brz10} (number of left quotients) of the language. 
The \emph{syntactic complexity of a class} of regular languages is the maximal syntactic complexity of languages in that class expressed as a function of the quotient complexity~$n$.

If $w=uxv$ for some $u,v,x\in\Sigma^*$, then  $u$ is a \emph{prefix} of $w$, $v$ is a \emph{suffix} of $w$ and  
 $x$ is a {\em factor\/} of $w$.
Prefixes and suffixes of $w$ are also factors of $w$.
A~language $L$ is \emph{prefix-free} (respectively, \emph{suffix-free, factor-free}) if $w, u \in L$ and $u$ is a prefix (respectively, \emph{suffix, factor}) of $w$, then  $u=w$.
A language is \emph{bifix-free} if it is both prefix- and suffix-free.
These languages play an important role in coding theory, have applications in such areas as cryptography, data compression, and information transmission, and have been studied extensively; see~\cite{BPR09} for example. 
In particular, suffix-free languages (with the exception of $\{\eps\}$, where $\eps$ is the empty word) 
are suffix codes.
Moreover, suffix-free languages are special cases of suffix-convex languages, where a language is \emph{suffix-convex} if it satisfies the condition that, if a word $w$ and its suffix $u$ are in the language, then so is every suffix of $w$ that has $u$ as a  suffix~\cite{AnBr09,Thi73}. 
We are interested only in regular suffix-free languages.

The syntactic complexity of prefix-free languages was proved to be $n^{n-2}$ in~\cite{BLY12}.
The syntactic complexities of suffix-, bifix-, and factor-free languages were also studied in~\cite{BLY12}, and the following lower bounds were established
$(n-1)^{n-2}+n-2$, 
$(n-1)^{n-3}+(n-2)^{n-3}+ (n-3)2^{n-3}$, and $(n-1)^{n-3}+ (n-3)2^{n-3}+1$, respectively.
It was conjectured that these bounds are also upper bounds; we prove the conjecture for suffix-free languages in this paper.

A much abbreviated version of these results appeared in~\cite{BrSz15}.

\section{Preliminaries}

\subsection{Languages, automata and transformations}
Let $\Sig$ be a finite, non-empty alphabet and let $L\subseteq \Sig^*$ be a language.
The \emph{left quotient} or simply \emph{quotient} of a language $L$ by a word $w\in\Sig^*$ is denoted by $L.w$ and defined by $L.w=\{x\mid wx\in L\}$.
A language is regular if and only if it has a finite number of quotients.
We denote the set of quotients by $K=\{K_0,\dots,K_{n-1}\}$, where $K_0=L=L.\eps$ by convention.
Each quotient $K_i$ can be represented also as $L.w_i$, where $w_i\in\Sig^*$ is such that
$L.w_i=K_i$.

A \emph{deterministic finite automaton (DFA)} is a quintuple
$\cD=(Q, \Sigma, \delta, q_0,F)$, where
$Q$ is a finite non-empty set of \emph{states},
$\Sig$ is a finite non-empty \emph{alphabet},
$\delta\colon Q\times \Sig\to Q$ is the \emph{transition function},
$q_0\in Q$ is the \emph{initial} state, and
$F\subseteq Q$ is the set of \emph{final} states.
We extend $\delta$ to a function $\delta\colon Q\times \Sig^*\to Q$ as usual.

The \emph{quotient DFA} of a regular language $L$ with $n$ quotients is defined by
$\cD=(K, \Sigma, \delta_\cD, K_0,F_\cD)$, where 
$\delta_\cD(K_i,w)=K_j$ if and only if $K_i.w=K_j$, 
and $F_\cD=\{K_i\mid \eps \in K_i\}$.
To simplify the notation, without loss of generality we use the set $Q=\{0,\dots,n-1\}$ of subscripts of quotients as the set of states of $\cD$; then $\cD$ is denoted by
$\cD=(Q, \Sigma, \delta, 0,F)$, where $\delta(i,w)=j$  if $\delta_\cD(K_i,w)=K_j$, and $F$ is the set of subscripts of quotients in $F_\cD$. The quotient corresponding to $q\in Q$ is then $K_q=\{w\mid \delta_\cD(K_q,w)\in F_\cD\}$.
The quotient $K_0=L$ is the \emph{initial} quotient. A quotient is \emph{final} if it contains $\eps$.
A state $q$ is \emph{empty} if its quotient $K_q$ is empty.

The quotient DFA of $L$ is isomorphic to each complete minimal DFA of $L$.
The number of states in the quotient DFA of $L$ (the quotient complexity of $L$) is therefore equal to the state complexity~of~$L$.

In any DFA, each letter $a\in \Sig$ induces a transformation of the set $Q$ of $n$ states.
Let $\cT_{Q}$ be the set of all $n^n$ transformations of $Q$; then $\cT_{Q}$ is a monoid under composition. 
The \emph{image} of $q\in Q$ under transformation $t$ is denoted by $qt$.
If $s,t$ are transformations of $Q$, their composition is denoted $s\circ t$ and defined by
$q(s \circ t)=(qs)t$; the $\circ$ is usually omitted.
The \emph{in-degree} of a state $q$ in a transformation $t$ is the cardinality of the set $\{p \mid pt=q\}$.

The \emph{identity} transformation $\tid$ maps each element to itself.
For $k\ge 2$, a transformation (permutation) $t$ of a set $P=\{q_0,q_1,\ldots,q_{k-1}\} \subseteq Q$ is a \emph{$k$-cycle}
if $q_0t=q_1, q_1t=q_2,\ldots,q_{k-2}t=q_{k-1},q_{k-1}t=q_0$.
A $k$-cycle is denoted by $(q_0,q_1,\ldots,q_{k-1})$.
If a transformation $t$ of $Q$ is a $k$-cycle of some $P \subseteq Q$, we say that $t$ \emph{has a $k$-cycle}.
A~transformation \emph{has a cycle} if it has a $k$-cycle for some $k\ge 2$.
A~2-cycle $(q_0,q_1)$ is called a \emph{transposition}.
A transformation is \emph{unitary} if it changes only one state $p$ to a state $q\neq p$; it is denoted by $(p\to q)$.
A transformation is \emph{constant} if it maps all states to a single state $q$; it is denoted by $(Q\to q)$.

The binary relation $\omega_t$ on $Q \times Q$ is defined as follows: For any $i, j \in Q$, $i \mathbin{\omega_t} j$ if and only if $it^k = jt^\ell$ for some $k, \ell \ge 0$. This is an equivalence relation, and each equivalence class is called an \emph{orbit}~\cite{GaMa09} of $t$. For any $i \in Q$, the orbit of $t$ containing $i$ is denoted by $\omega_t(i)$. 
An orbit contains either exactly once cycle and no fixed points or exactly one fixed point and no cycles. The~set of all orbits of $t$  is a partition of $Q$.

If $w \in \Sig^*$  induces a transformation $t$, we denote this by 
$w\colon t$.
A~transformation mapping $i$ to $q_i$ for $i=0, \dots, n-1$ is sometimes denoted by
$[q_0, \dots,q_{n-1}]$. 
By a slight abuse of notation we sometimes represent the transformation $t$ induced by $w$ by $w$ itself, and write $qw$ instead of $qt$.

The \emph{transition semigroup} of a DFA $\cD=(Q,\Sig,\delta,0, F)$ is the semigroup of transformations
of $Q$ generated by the transformations induced by the letters of $\Sig$.
Since the transition semigroup of a minimal DFA of a language $L$ is isomorphic to the 
syntactic semigroup of $L$~\cite{Pin97}, syntactic complexity is equal to the cardinality 
of the transition semigroup.

\subsection{Suffix-free languages}

For any transformation $t$, consider the sequence $(0,0t,0t^2,\dots)$; we call it the \emph{$0$-path} of $t$.
Since $Q$ is finite, there exist $i,j$ such that $0,0t,\dots,$ $0t^i,0t^{i+1},$ $\dots,0t^{j-1}$ are distinct but $0t^j=0t^i$.
The integer $j-i$ is the \emph{period} of $t$ and if $j-i=1$, $t$ is \emph{initially aperiodic}.

Let $Q=\{0,\ldots,n-1\}$, 
let $\cD_n=(Q, \Sigma, \delta, 0,F)$ be a minimal DFA accepting a language $L$, and let $T(n)$ be its transition semigroup.
The following observations are well known~\cite{BLY12,HaSa09}:

\begin{lemma}
\label{lem:sf}
If $L$ is a suffix-free language, then
\be
\item
There exists $w \in \Sig^*$ such that $L.w=\emp$; hence $\cD_n$ has an empty state, which is state $n-1$ by convention.
\item
For $w,x\in \Sig^+$, if $L.w\neq \emp$, then $L.w\neq L.xw$.
\item
If $L.w\neq \emp$, then $L.w=L$ implies $w=\eps$.
\item
For any $t\in T(n)$, the 0-path of $t$ in $\cD_n$ is aperiodic and ends in $n-1$. 
\ee
\end{lemma}

An (unordered) pair $\{p,q\}$ of distinct states in $Q \setminus \{0,n-1\}$ is \emph{colliding} (or $p$ \emph{collides} with $q$) in $T(n)$ if there is a transformation $t \in T(n)$ such that $0t = p$ and $rt = q$ for some $r \in Q \setminus \{0,n-1\}$. 
A pair of states is \emph{focused} by a transformation $u$ of $Q$ if $u$ maps both states of the pair to a single state $r \not\in \{0,n-1\}$. We then say that $\{p,q\}$ is \emph{focused to state $r$}.
If $L$ is a suffix-free language, then from Lemma~\ref{lem:sf}~(2) it follows that if $\{p,q\}$ is colliding in $T(n)$, there is no transformation $t' \in T(n)$ that focuses $\{p,q\}$.
So colliding states can be mapped to a single state by a transformation in $T(n)$ only if that state is the empty state $n-1$.

\begin{remark}
\label{rem:suffix_small}
If $n=1$, the only suffix-free language is the empty language $\emp$ and $\sig(\emp)=1$.
If $n\ge 2$ and $\Sig=\{a\}$, the language $L=a^{n-2}$ is the only suffix-free language of quotient complexity $n$, and its syntactic complexity is $\sig(L)=n-1$.

Assume now that $|\Sig|\ge 2$.
If $n=2$, the language $L=\eps$ is the only suffix-free language, and $\sig(L)=1$.
If $n=3$,  the tight upper bound on syntactic complexity of suffix-free languages is 3, and the language $L=ab^*$ over $\Sig=\{a,b\}$ meets this bound~\textrm{\cite{BLY12}}.
\qedb
\end{remark}

\subsection{Suffix-free semigroups}

Since the cases where $2\le n\le 3$ were easily resolved, we assume now that $n\ge 4$.
Also, without loss of generality, we assume that $Q=\{0,\ldots,n-1\}$.
A transformation of $Q$ is \emph{suffix-free} if it can belong to the transition semigroup of a minimal DFA accepting a suffix-free language. 

The following set of all suffix-free transformations was defined in~\cite{BLY12}:
For $n \ge 2$ let
\begin{eqnarray*}
\Bsf(n) = \{t \in \cT_Q & \mid & 0 \not\in Qt\text{,  } \, (n-1)t=n-1\text{, }  \txt{and for all} j\ge 1,\\
& & 0t^j = n-1 \txt{or} 0t^j \neq qt^j ~~\forall q, \; 0 < q < n-1\}.
\end{eqnarray*}

Suppose $\cD$ is a minimal DFA accepting a non-empty suffix-free language $L$.

Each transformation in $\Bsf(n)$ satisfies three conditions.
First, no transformation $t$ induced by a word $y$ can map any state $q$ to 0, because then we would have some words $x$ and $z$ such that $0x=q$, $qy=0$, $z\in L$ and $xyz\in L$.
Second, since $\cD$ must have an empty state $n-1$, that state must be mapped to itself by every transformation. 
Third, since state $0t^j$ collides with $qt^j$ for all $j\ge 1$ if $0t^j\neq n-1$, these two states cannot be focused.

It was shown in~\cite{BLY12} that if $\cD$ is a minimal DFA accepting a suffix-free language, then its syntactic semigroup is  contained in $\mathrm{\mathbf{B}_{sf}}(n)$. 
Since the set $\mathrm{\mathbf{B}_{sf}}(n)$ is not a semigroup, we look for largest semigroups contained in $\mathrm{\mathbf{B}_{sf}}(n)$.
Two such semigroups were introduced in~\cite{BLY12}.

The first semigroup is defined as follows: For $n \ge 4$, let
\begin{eqnarray*}
\Vsf(n) = \{t \in \Bsf(n) &\mid& \txt{for all} p, q \in Q \txt{where} p \neq q, \\ 
&& \txt{we have} pt = qt = n-1 \txt{or} pt \neq qt\}.
\end{eqnarray*}

It was shown in~\cite{BrSz15a} that for $n\ge 4$, semigroup $\Vsf(n)$ is generated by the following set of transformations of $Q$:

\bi
\item
${a}\colon (0\to n-1)(1,\dots,n-2)$,
\item
${b}\colon (0\to n-1)(1,2)$, 
\item
for $1\le p \le n-2$, $ c_p\colon (p\to n-1)(0\to p)$.
\ei

\begin{proposition}[{\cite[Proposition 3]{BrSz15a}}]\label{prop:unique_all_colliding}
For $n\ge 4$, $\Vsf(n)$ is the unique maximal semigroup of a suffix-free language in which every two states from $\{1,\dots,n-2\}$ are colliding.
\end{proposition}
\begin{proof}
For each pair $p,q \in Q \setminus \{0,n-1\}$, $p\neq q$, there is a transformation $t \in \Vsf(n)$ with $0t = p$ and $qt = q$.
Thus all pairs are colliding.
If all pairs are colliding, then for each $p,q \in Q \setminus \{n-1\}$, there is no transformation $t$ with $pt = qt \neq n-1$, for this would violate suffix-freeness.
By definition, $\Vsf(n)$ has all other transformations that are possible for a suffix-free language, and hence is unique.
\end{proof}

The second semigroup from~\cite{BLY12} was defined as follows:
For $n \ge 4$, let 
$$\Wsf(n) = \{ t \in \Bsf(n) \mid 0t = n-1 \txt{or} qt = n-1~~\forall~q, \; 1 \le q \le n-2 \}.$$

The transition semigroup $\Wsf(n)$ has cardinality $(n-1)^{n-2}+n-2$, which is a lower bound on the complexity of suffix-free languages established in~\cite{BLY12} using a witness DFA with an alphabet with $n+2$ letters.
Our first contribution is to simplify the witness of~\cite{BLY12} with the transition semigroup $\Wsf(n)$ by using an alphabet with only five letters, as stated in Definition~\ref{def:sf_wit}. The transitions induced by inputs $a$, $b$, $c$, and $e$ are the same as in~\cite{BLY12}.

\begin{definition}[Suffix-Free Witness]
\label{def:sf_wit}
For $n\ge 4$, 
we define the DFA 
$\cW_n =(Q,\Sig_\cW,\delta_\cW,0,\{1\}),$
where $Q=\{0,\ldots,n-1\}$, $\Sig_\cW=\{a,b,c,d,e\}$, 
and $\delta_\cW$ is defined by the transformations
$a\colon (0\to n-1) (1,\ldots,n-2)$,
$b\colon (0\to n-1) (1,2)$,
$c\colon (0\to n-1) (n-2\to 1)$,
$d\colon (\{0,1\}\to n-1)$,
and $e\colon (Q\setminus \{0\} \to n-1)(0\to 1)$.
For $n=4$, $a$ and $b$ coincide, and we can use $\Sig_\cW=\{b,c,d,e\}$.
\end{definition}
\begin{figure}[ht]
\unitlength 10pt
\begin{center}\begin{picture}(27,13)(0,0)
\gasset{Nh=2.5,Nw=2.5,Nmr=1.25,ELdist=0.3,loopdiam=1.5}
\node(0)(1,7){0}\imark(0)
\node(1)(9,7){1}\rmark(1)
\node(2)(17,7){2}
\node(3)(25,7){3}
\node(4)(13,1){4}
\drawedge(0,1){$e$}
\drawedge[curvedepth=1](1,2){$a,b$}
\drawloop(1){$c$}
\drawedge[curvedepth=1,ELdist=-1](2,1){$b$}
\drawloop(2){$c,d$}
\drawedge(2,3){$a$}
\drawedge[curvedepth=-5,ELdist=-1](3,1){$a,c$}
\drawloop(3){$b,d$}
\drawedge[ELdist=-2.5](0,4){$\Sigma\setminus\{e\}$}
\drawedge(1,4){$d,e$}
\drawedge(2,4){$e$}
\drawedge(3,4){$e$}
\drawloop[loopangle=270,ELpos=25](4){$\Sigma$}
\end{picture}\end{center}
\caption{DFA $\cW_5$.}
\label{fig:witness}
\end{figure}

The structure of $\cW_n$ is illustrated in Figure~\ref{fig:witness} for $n=5$.
We claim that no pair of states from $Q$ is colliding in $S_n$.
If $0t=p\not\in \{0,n-1\}$, then $t$ is not the identity but must be induced by a word of the form $ew$ for some $w\in\Sig^*$. 
Such a word maps every $r\not\in \{0,n-1\}$ to $n-1$.
Hence $q=rt=n-1$, and $p$ and $q$ do not collide.

\begin{proposition}
\label{prop:sf_wit}
For $n\ge 4$ the DFA $\cW_n$ of Definition~\ref{def:sf_wit} is minimal, suffix-free, and its transition semigroup is $\Wsf(n)$ with cardinality $(n-1)^{n-2}+n-2$. In particular, $\Wsf(n)$ contains
(a) all $(n-1)^{n-2}$ transformations that send 0 and $n-1$ to $n-1$ and map 
$Q\setminus \{0,n-1\}$ to $Q\setminus \{0\}$, and 
(b) all $n-2$ transformations that send 0 to a state in $Q\setminus \{0,n-1\}$ and map all the other states to $n-1$.
\end{proposition}
\begin{proof}
The proof of minimality and suffix-freeness is the same as in~\cite{BLY12}.
Note that $\Wsf(n)$ contains only transformations of type (a) and (b) and it contains all such transformations.
The transformations induced by $a$, $b$, and $c$ restricted to 
$Q\setminus \{0,n-1\}$ generate all transformations of the middle $n-2$ states. 
Together with $d$, they generate all transformations of $Q$ that send 0 to $n-1$, fix $n-1$, and send a state in $Q\setminus\{0, n-1\}$ to $Q\setminus \{0\}$. 
Also, for $0\le i\le n-3$, words $ea^i$ send $0$ to some $p\not\in \{0,n-1\}$ and map all the other states to $n-1$.
Hence the transition semigroup of $\cW_n$ is $\Wsf(n)$ and the proposition holds.
\end{proof}

\begin{proposition}\label{prop:unique_no_colliding}
For $n \ge 4$, $\Wsf(n)$ is the unique maximal semigroup of a suffix-free language in which no two states from $\{1,\dots,n-2\}$ are colliding.
\end{proposition}
\begin{proof}
Since $0t = n-1$ or $qt = n-1$ for all $q \in Q \setminus \{0\}$, no two states are colliding.
%
By Proposition~\ref{prop:sf_wit} all such transformations are in $\Wsf(n)$.
\end{proof}

If $n=4$ and $n=5$, the tight upper bounds on the size of suffix-free transition semigroups are 13, and 73, respectively~\cite{BLY12}, and these bounds are met by $\Vsf(n)$.
It was shown in~\cite{BLY12} that there is a suffix-free witness DFA with $n$ states and an alphabet of size $n+2$ that meets the bound $(n-1)^{n-2}+n-2$ for $n\ge 4$, and the transition semigroup of this DFA is $\Wsf(n)$. For $n=4$ and $n=5$, these bounds are 11 and 67, and hence are smaller than the sizes of $\Vsf(4)$ and $\Vsf(5)$. However, for $n\ge 6$, $(n-1)^{n-2}+n-2$ is the largest known lower bound, and it is met by $\Wsf(n)$.
The remainder of this paper is devoted to proving that $(n-1)^{n-2}+n-2$ is also an upper bound.

\section{Upper bound for suffix-free languages}

Our main result shows that the lower bound $(n-1)^{n-2}+n-2$ on the syntactic complexity of suffix-free languages is also an upper bound for $n \ge 7$.
Our approach is as follows: We consider a minimal DFA  $\cD_n=(Q, \Sigma, \delta, 0,F)$, where $Q=\{0,\ldots,n-1\}$, of an arbitrary suffix-free language with $n$ quotients and let $T(n)$ be the transition semigroup of $\cD_n$. 
We also deal with the witness DFA $\cW_n =(Q,\Sig_\cW,\delta_\cW,0,\{1\})$ of Definition~\ref{def:sf_wit} that has the same state set as $\cD_n$ and whose transition semigroup is $\Wsf(n)$. We shall show that there is an injective mapping $\varphi \colon T(n)\to\Wsf(n)$, and this will prove that $|T(n)|\le |\Wsf(n)|$.

A note about terminology may be helpful to the reader. The semigroups $T(n)$ and $\Wsf(n)$ share the set $Q$. When we say that a pair of states from $Q$ is \emph{colliding} we mean that it is colliding in $T(n)$; there is no room for confusion because no pair of states is colliding in $\Wsf(n)$. Since we are dealing with suffix-free languages, a transformation that focuses a colliding pair cannot belong to $T(n)$.

\begin{theorem}[Suffix-Free Languages]\label{thm:suffix-free_upper-bound}
For $n\ge 6$ the syntactic complexity of the class of suffix-free languages with $n$ quotients is $(n-1)^{n-2}+n-2$.
\end{theorem}

\begin{proof}
The case $n=6$ has been proved in~\cite{BLY12}; hence assume that $n\ge 7$.
In~\cite{BLY12} and in Proposition~\ref{prop:sf_wit} it was  shown that $(n-1)^{n-2}+n-2$ is a lower bound for $n\ge 7$; hence it remains to prove that it is also an upper bound, and we do this here.
We have the following cases:

\noindent{\hglue 15pt}
{\bf Case 1:} $t \in \Wsf(n)$.\\
Let $\varphi(t) = t$; obviously $\varphi$ is injective.


\noindent{\hglue 15pt}
{\bf Case 2:} $t \not\in \Wsf(n)$, and $t$ has a cycle.\\
In this case and in all the following ten cases let $p = 0t$. By Lemma~\ref{lem:sf} (4) we have the chain
$$0 \stackrel{t}{\rightarrow} p \stackrel{t}{\rightarrow} pt \stackrel{t}{\rightarrow} \dots \stackrel{t}{\rightarrow} pt^k \stackrel{t}{\rightarrow} n-1,$$
where $k \ge 0$. Observe that pairs $\{pt^i,pt^j\}$ for $0 \le i<j \le k$ are colliding, since transformation $t^{i+1}$ maps $0$ to $pt^i$ and $pt^{j-i-1}$ to $pt^j$.
Also, $p$ collides with any state from a cycle of $t$ and any fixed point of $t$ other than $n-1$.

Let $r$ be minimal among the states that appear in cycles of $t$, that is,
$$r = \min\{q\in Q \mid q\text{ is in a cycle of }t\}.$$
Let $s$ be the transformation illustrated in Fig.~\ref{fig:case2} and defined by
\begin{center}
$0 s = n-1, \quad p s = r, \quad (pt^i) s = pt^{i-1} \text{ for } 1 \le i \le k,$

$q s = q t \text{ for the other states } q \in Q.$
\end{center}
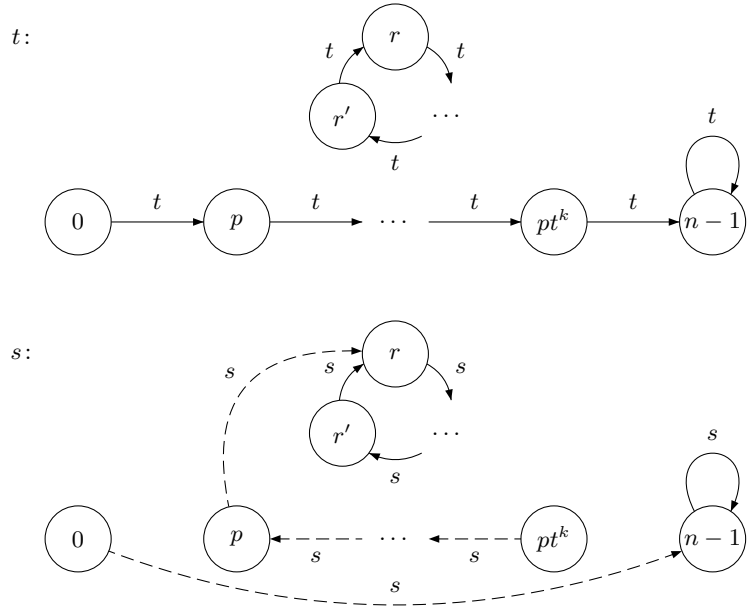
\begin{figure}[ht]
\unitlength 10pt\footnotesize
\begin{center}\begin{picture}(28,24)(0,-1)
\gasset{Nh=2.5,Nw=2.5,Nmr=1.25,ELdist=0.5,loopdiam=2}
\node[Nframe=n](name)(0,21){$t\colon$}
\node(0)(2,14){0}
\node(p)(8,14){$p$}
\node[Nframe=n](pdots)(14,14){$\dots$}
\node(pt^k)(20,14){$pt^k$}
\node(n-1)(26,14){$n-1$}
\node(rt^-1)(12,18){$r'$}
\node(r)(14,21){$r$}
\node[Nframe=n](rdots)(16,18){$\dots$}
\drawedge(0,p){$t$}
\drawedge(p,pdots){$t$}
\drawedge(pdots,pt^k){$t$}
\drawedge(pt^k,n-1){$t$}
\drawloop[loopangle=90](n-1){$t$}
\drawedge[curvedepth=1](rt^-1,r){$t$}
\drawedge[curvedepth=1](r,rdots){$t$}
\drawedge[curvedepth=1](rdots,rt^-1){$t$}

\node[Nframe=n](name)(0,9){$s\colon$}
\node(0')(2,2){0}
\node(p')(8,2){$p$}
\node[Nframe=n](pdots')(14,2){$\dots$}
\node(pt^k')(20,2){$pt^k$}
\node(n-1')(26,2){$n-1$}
\node(rt^-1')(12,6){$r'$}
\node(r')(14,9){$r$}
\node[Nframe=n](rdots')(16,6){$\dots$}
\drawedge[curvedepth=-2.5,dash={.5 .25}{.25}](0',n-1'){$s$}
\drawedge[dash={.5 .25}{.25}](pdots',p'){$s$}
\drawedge[dash={.5 .25}{.25}](pt^k',pdots'){$s$}
\drawedge[curvedepth=3.5,dash={.5 .25}{.25}](p',r'){$s$}
\drawedge[curvedepth=1](rt^-1',r'){$s$}
\drawedge[curvedepth=1](r',rdots'){$s$}
\drawedge[curvedepth=1](rdots',rt^-1'){$s$}
\drawloop[loopangle=90](n-1'){$s$}
\end{picture}\end{center}
\caption{Case~2 in the proof of Theorem~\ref{thm:suffix-free_upper-bound}.}
\label{fig:case2}
\end{figure}

By Proposition~\ref{prop:sf_wit},  $\varphi(t)=s$ is in $\Wsf(n)$, since it maps 0 to $n-1$,  fixes $n-1$, and does not map any states to 0.
 Note that the sets of cyclic states in both $t$ and $s$ are the same.
Let $r'$ be the state from the cycle  of $t$  such that $r't = r$; then transformation $s$ has the following properties:
\bd
\item[(a)] 
Since $p$ collides with any state in a cycle of $t$, $\{p,r'\}$ is a colliding pair focused by $s$ to state $r$ in the cycle. 
Moreover, if $q'$ is a state in a cycle of $s$, and $\{q,q'\}$ is colliding and focused by $s$ to a state in a cycle, then that state must be $r$ (the minimal state in the cycles of $s$), $q$ must be $p$, and $q'$ must be $r'$.

This follows from the definition of $s$. Since $s$ differs from $t$ only in the mapping of states $pt^i$ and $0$, any colliding pair focused by $s$ contains $pt^i$ for some $i$, $0 \le i \le k$. Only $p$ is mapped to $r$, which is in a cycle of $t$, and $r'$ is the only state in that cycle that is mapped to $r$.

\item[(b]) For each $i$ with $1 \le i < k$, there is precisely one state $q$ colliding with $pt^{i-1}$ and mapped by $s$ to $pt^i$, and that state is $q=pt^{i+1}$.

Clearly $q=pt^{i+1}$ satisfies this condition. Suppose that $q \neq pt^{i+1}$. Since $pt^{i+1}$ is the only state mapped to $pt^i$ by $s$ and not by $t$, it follows that $qt = qs = pt^i$. So $q$ and $pt^{i-1}$ are focused to $pt^i$ by $t$; since they collide, this is a contradiction.

\item[(c)] Every focused colliding pair consists of states from the orbit of $p$.

This follows from the fact that all the states except $0$ that are mapped by $s$ differently than by $t$ belong to the orbit of $p$.

\item[(d)] $s$ has a cycle.
\ed
\smallskip

From~(a), $s \not\in T(n)$ and so  $s$ is different from the transformations of Case~1.
\smallskip

Given a transformation $s$ from this case we will construct a unique $t$ that results in $s$ when the definition of $s$ given above is applied. This will show that our mapping $\varphi$ has an inverse, and so is injective.
From~(a) there is the unique colliding pair focused to a state in a cycle. 
Moreover, one its states, say $p$, is not in this cycle and another one, say $r'$, is in this cycle. It follows that $0t = p$. Since there is no state $q \neq 0$ such that $qt=p$, the only state mapped to $p$ by $s$ is $pt$. From~(b) for $i=1,\ldots,k-1$ state $pt^{i+1}$ is uniquely determined. Finally, for $i=k$ there is no state colliding with $pt^{k-1}$ and mapped to $pt^k$; so $pt^{k+1} = n-1$. Since the other transitions in $s$ are defined exactly as in $t$, this procedure defines the inverse function $\varphi^{-1}$ for the transformations of this case.


\noindent{\hglue 15pt}
{\bf Case 3:} $t \not\in \Wsf(n)$, $t$ has no cycles, but $pt \neq n-1$.\\
Let $s$ be the transformation illustrated in Fig.~\ref{fig:case3} and defined by
\begin{center}
$0 s = n-1, \quad p s = p, \quad (pt^i) s = pt^{i-1} \text{ for } 1 \le i \le k,$

$q s = q t \text{ for the other states } q \in Q.$
\end{center}
\begin{figure}[ht]
\unitlength 10pt\footnotesize
\begin{center}\begin{picture}(28,14)(0,-1)
\gasset{Nh=2.5,Nw=2.5,Nmr=1.25,ELdist=0.5,loopdiam=2}
\node[Nframe=n](name)(0,11){$t\colon$}
\node(0)(2,9){0}
\node(p)(8,9){$p$}
\node[Nframe=n](pdots)(14,9){$\dots$}
\node(pt^k)(20,9){$pt^k$}
\node(n-1)(26,9){$n-1$}
\drawedge(0,p){$t$}
\drawedge(p,pdots){$t$}
\drawedge(pdots,pt^k){$t$}
\drawedge(pt^k,n-1){$t$}
\drawloop[loopangle=90](n-1){$t$}

\node[Nframe=n](name)(0,4){$s\colon$}
\node(0')(2,2){0}
\node(p')(8,2){$p$}
\node[Nframe=n](pdots')(14,2){$\dots$}
\node(pt^k')(20,2){$pt^k$}
\node(n-1')(26,2){$n-1$}
\drawedge[curvedepth=-2.5,dash={.5 .25}{.25}](0',n-1'){$s$}
\drawedge[dash={.5 .25}{.25}](pdots',p'){$s$}
\drawedge[dash={.5 .25}{.25}](pt^k',pdots'){$s$}
\drawloop[loopangle=90](p'){$s$}
\drawloop[loopangle=90](n-1'){$s$}
\end{picture}\end{center}
\caption{Case~3 in the proof of Theorem~\ref{thm:suffix-free_upper-bound}.}
\label{fig:case3}
\end{figure}
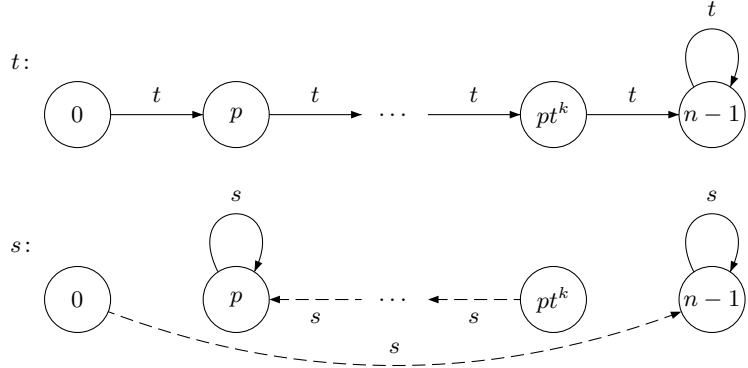

Observe that $s$ has the following properties:
\bd
\item[(a)] $\{p,pt\}$ is the only colliding pair focused by $s$ to a fixed point. Moreover the fixed point is contained in the pair, and has in-degree 2.

This follows from the definition of $s$, since any colliding pair focused by $s$ contains $pt^i$ ($0 \le i \le k$), and only $pt$ is mapped to $p$, which is a fixed point.
Also, no state except $0$ can be mapped to $p$ by $t$ because this would violate suffix-freeness; so only $p$ and $pt$ are mapped by $s$ to $p$, and $p$ has in-degree 2.

\item[(b)] For each $i$ with $1 \le i < k$, there is precisely one state $q$ colliding with $pt^{i-1}$ and mapped to $pt^i$, and that state is $q=pt^{i+1}$.

This follows exactly like Property~(b) from Case~2.

\item[(c)] Every colliding pair focused by $s$ consists of states from the orbit of $p$.

This follows exactly like Property~(c) from Case~2.

\item[(d)] $s$ does not have a cycle, but has a fixed point $f \neq n-1$ with in-degree $\ge 2$, which is $p$.
\ed

From~(a), $s \not\in T(n)$ and so $s$ is different from the transformations of Case~1. Here $s$ does not have a cycle in contrast with the transformations of Case~2.

As before, $s$ uniquely defines the transformation $t$ from which it is obtained: From~(a) there is the unique colliding pair $\{p,pt\}$ focused to the fixed point $p$. Thus $0t = p$. 
Then, as in Case~2, for $i=1,\ldots,k-1$ state $pt^{i+1}$ is uniquely defined, and $pt^k = n-1$. Since the other transitions in $s$ are defined exactly as in $t$, this procedure yields the inverse function $\varphi^{-1}$ for  this case.

\noindent{\hglue 15pt}
{\bf Case 4:} $t$ does not fit in any of the previous cases, but there is a fixed point $r \in Q \setminus \{0,n-1\}$ with in-degree $\ge 2$.\\
Let $s$ be the transformation illustrated in Fig.~\ref{fig:case4} and defined by
\begin{center}
$0 s = n-1, \quad p s = r,$

$q s = q t \text{ for the other states } q \in Q.$
\end{center}
\begin{figure}[ht]
\unitlength 10pt\footnotesize
\begin{center}\begin{picture}(28,19)(0,-1)
\gasset{Nh=2.5,Nw=2.5,Nmr=1.25,ELdist=0.5,loopdiam=2}
\node[Nframe=n](name)(0,16){$t\colon$}
\node(0)(2,11){0}
\node(p)(14,11){$p$}
\node(n-1)(26,11){$n-1$}
\node(q)(9,14){}
\node(r)(19,14){$r$}
\drawedge(0,p){$t$}
\drawedge(p,n-1){$t$}
\drawloop[loopangle=90](n-1){$t$}
\drawedge(q,r){$t$}
\drawloop[loopangle=90](r){$t$}

\node[Nframe=n](name)(0,7){$s\colon$}
\node(0')(2,2){0}
\node(p')(14,2){$p$}
\node(n-1')(26,2){$n-1$}
\node(q')(9,5){}
\node(r')(19,5){$r$}
\drawedge[curvedepth=-2.5,dash={.5 .25}{.25},ELdist=-1](0',n-1'){$s$}
\drawedge[dash={.5 .25}{.25},curvedepth=-0.5](p',r'){$s$}
\drawloop[loopangle=90](n-1'){$s$}
\drawedge(q',r'){$s$}
\drawloop[loopangle=90](r'){$s$}
\end{picture}\end{center}
\caption{Case~4 in the proof of Theorem~\ref{thm:suffix-free_upper-bound}.}
\label{fig:case4}
\end{figure}
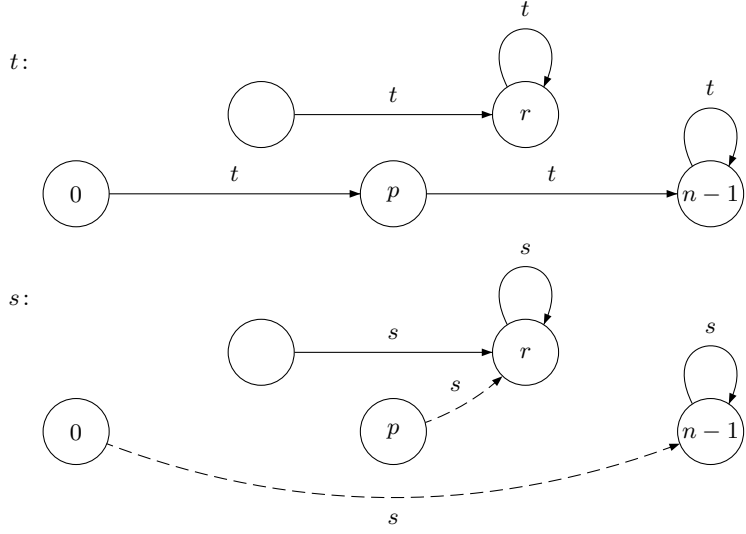

Observe that $s$ has the following properties:
\bd
\item[(a)] $\{p,r\}$ is the only colliding pair focused by $s$ to a fixed point, where the fixed point is contained in the pair. Moreover the fixed point has in-degree at least 3.

Since $s$ differs from $t$ only by the mapping of states $0$ and $p$, it follows that all focused colliding pairs contain $p$. Since $p$ is mapped to $r$, the second state in the pair must be the fixed point $r$. Since $r$ has in-degree at least $2$ in $t$, and $s$ additionally maps $p$ to $r$, $r$ has in-degree at least 3.

\item[(b)] $s$ does not have a cycle, but has a fixed point $\neq n-1$ with in-degree $\ge 3$, which is $r$.
\ed

From~(a) we have $s \not\in T(n)$, and so $s$ is different from the transformations of Case~1. Here $s$ does not have a cycle in contrast with the transformations of Case~2. Also from~(a) we know that the fixed point in the distinguished colliding pair has in-degree $\ge 3$, whereas in Case~3 it has  in-degree~2.

From~(a) we see that the colliding pair $\{p,r\}$ in which $r$ is a fixed point and $p$ is not is uniquely defined.
Hence $0t=p$ and $pt=n-1$, and $t$ is again uniquely defined from $s$.

\noindent{\hglue 15pt}
{\bf Case 5:} $t$ does not fit in any of the previous cases, but there is a state $r$ with in-degree $\ge 1$ that is not a fixed point and satisfies $rt \neq n-1$.

Since there are no fixed points in $s$ with in-degree $\ge 2$ other than $n-1$, and there are no cycles, it follows that $r$ belongs to the orbit of $n-1$. Hence we can choose $r$ such that $rt \neq n-1$ and $rt^2 = n-1$.

Let $s$ be the transformation illustrated in Fig.~\ref{fig:case5} and defined by
\begin{center}
$0 s = n-1, \quad p s = r t,$

$q s = q t \text{ for the other states } q \in Q.$
\end{center}
\begin{figure}[ht]
\unitlength 10pt\footnotesize
\begin{center}\begin{picture}(28,19)(0,-1)
\gasset{Nh=2.5,Nw=2.5,Nmr=1.25,ELdist=0.5,loopdiam=2}
\node[Nframe=n](name)(0,16){$t\colon$}
\node(0)(2,11){0}
\node(p)(14,11){$p$}
\node(n-1)(26,11){$n-1$}
\node(q)(11,14){}
\node(r)(17,14){$r$}
\node(rt)(22,14){$rt$}
\drawedge(0,p){$t$}
\drawedge(p,n-1){$t$}
\drawloop[loopangle=90](n-1){$t$}
\drawedge(q,r){$t$}
\drawedge(r,rt){$t$}
\drawedge(rt,n-1){$t$}

\node[Nframe=n](name)(0,7){$s\colon$}
\node(0')(2,2){0}
\node(p')(14,2){$p$}
\node(n-1')(26,2){$n-1$}
\node(q')(11,5){}
\node(r')(17,5){$r$}
\node(rt')(22,5){$rt$}
\drawedge[curvedepth=-2.5,dash={.5 .25}{.25},ELdist=-1](0',n-1'){$s$}
\drawedge[dash={.5 .25}{.25},ELdist=-1,curvedepth=-1](p',rt'){$s$}
\drawloop[loopangle=90](n-1'){$s$}
\drawedge(q',r'){$s$}
\drawedge(r',rt'){$s$}
\drawedge(rt',n-1'){$s$}
\end{picture}\end{center}
\caption{Case~5 in the proof of Theorem~\ref{thm:suffix-free_upper-bound}.}
\label{fig:case5}
\end{figure}
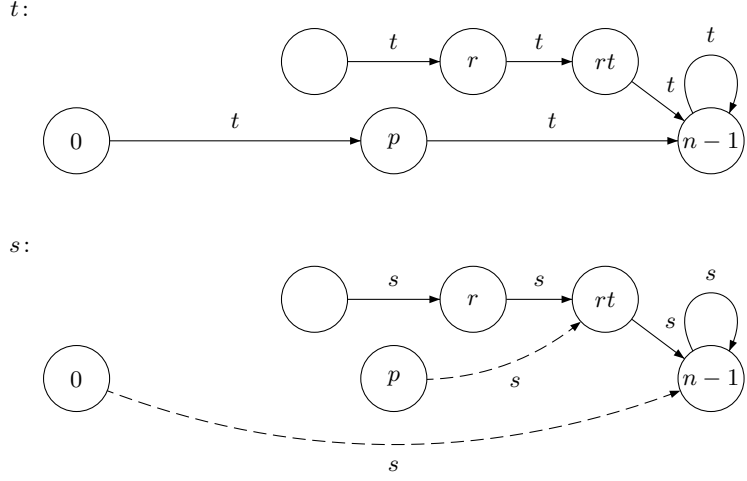

Observe that $s$ has the following properties:
\bd
\item[(a)] All focused colliding pairs contain $p$, and the second state from such a pair has in-degree $\ge 1$.

This follows since $s$ differs from $t$ only in the mapping of $0$ and $p$.

\item[(b)] The smallest $i$ with $ps^i = n-1$ is $2$.

\item[(c)] $s$ has neither a cycle nor a fixed point with in-degree $\ge 2$ other than $n-1$.
\ed

Note that $p$ and $r$ collide. Since $\{p,r\}$ is focused to $r t$, we have $s \not\in T(n)$ and so $s$ is different from the transformations of Case~1. Here $s$ does not have a cycle in contrast with the transformations of Case~2. Also $s$ does not have a fixed point other than $n-1$, and so is different from the transformations of Cases~3 and~4.

From~(a) all focused colliding pairs contain $p$. If there are two or more such pairs, $p$ is the only state in their intersection. If there is only one such pair, then it must be $\{p,r\}$, and $p$ is uniquely determined, since it has in-degree 0 and $r$ has in-degree $\ge 1$.
Hence $0t=p$ and $pt=n-1$, and again $t$ is uniquely defined from $s$.

\noindent{\hglue 15pt}
{\bf Case 6:} $t$ does not fit in any of the previous cases, but there is a state $r \in Q \setminus \{0,n-1\}$ with in-degree $\ge 2$.\\
Clearly $r \neq p$, since the in-degree of $p$ is 1. Also $rt = n-1$, as otherwise $t$ would fit in Case~5.

Let $R = \{r' \in Q \mid r't = r\}$; then $|R| \ge 2$.
We consider the following two sub-cases. If $p < r$, let $q_1$ be the smallest state in $R$ and let $q_2$ be the second smallest state; so $q_1 < q_2$. If $p > r$, let $q_1$ be the second smallest state in $R$, and let $q_2$ be the smallest state; so $q_2 < q_1$.

Let $s$ be the transformation illustrated in Fig.~\ref{fig:case6} and defined by
\begin{center}
$0 s = n-1, \quad p s = q_1, \quad r s = q_1, \quad q_1 s = q_2, \quad q_2 s = n-1,$

$q s = q t \text{ for the other states } q \in Q.$
\end{center}
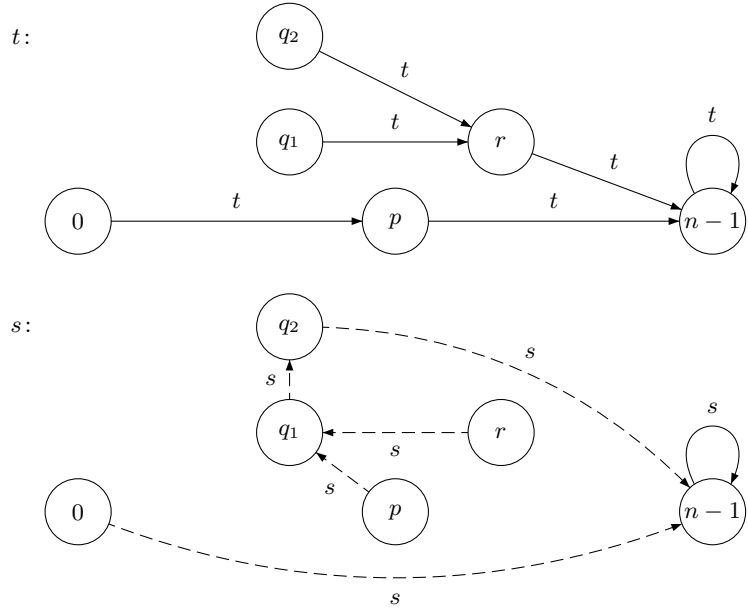
\begin{figure}[ht]
\unitlength 10pt\footnotesize
\begin{center}\begin{picture}(28,22)(0,-1)
\gasset{Nh=2.5,Nw=2.5,Nmr=1.25,ELdist=0.5,loopdiam=2}
\node[Nframe=n](name)(0,20){$t\colon$}
\node(0)(2,13){0}
\node(p)(14,13){$p$}
\node(n-1)(26,13){$n-1$}
\node(q1)(10,16){$q_1$}
\node(q2)(10,20){$q_2$}
\node(r)(18,16){$r$}
\drawedge(0,p){$t$}
\drawedge(p,n-1){$t$}
\drawloop[loopangle=90](n-1){$t$}
\drawedge(q1,r){$t$}
\drawedge(q2,r){$t$}
\drawedge(r,n-1){$t$}

\node[Nframe=n](name)(0,9){$s\colon$}
\node(0')(2,2){0}
\node(p')(14,2){$p$}
\node(n-1')(26,2){$n-1$}
\node(q1')(10,5){$q_1$}
\node(q2')(10,9){$q_2$}
\node(r')(18,5){$r$}
\drawedge[curvedepth=-2.5,dash={.5 .25}{.25},ELdist=-1](0',n-1'){$s$}
\drawloop[loopangle=90](n-1'){$s$}
\drawedge[dash={.5 .25}{.25}](p',q1'){$s$}
\drawedge[dash={.5 .25}{.25}](r',q1'){$s$}
\drawedge[dash={.5 .25}{.25}](q1',q2'){$s$}
\drawedge[dash={.5 .25}{.25},curvedepth=2](q2',n-1'){$s$}
\end{picture}\end{center}
\caption{Case~6 in the proof of Theorem~\ref{thm:suffix-free_upper-bound}.}
\label{fig:case6}
\end{figure}

Observe that $s$ has the following properties:
\bd
\item[(a)] There is only one focused colliding pair, namely $\{p,r\}$ mapped to $q_1$.

Clearly $p$ and $r$ collide. 
Note that no state can be mapped by $t$ to $q_1$ or $q_2$, since this would satisfy Case~5.
Because $q_1$ is the only state mapped by $s$ to $q_2$, it does not belong to a focused colliding pair. Also $0$ and $q_2$ are mapped to $n-1$. Since the other states are mapped exactly as in $t$, it follows that $s$ does not focus any other colliding pairs.

\item[(b)] The smallest $i$ with $ps^i = n-1$ is $3$.

\item[(c)] $s$ has neither a cycle nor a fixed point $\neq n-1$ with in-degree $\ge 2$.

This follows since $t$ does not have a cycle, and the states $0,p,r,q_1,q_2$ that are mapped differently by $s$ are in the orbit of $n-1$.
\ed

Since $s$ focuses the colliding pair $\{p,r\}$, $s$ is different from the transformations of Case~1. Also $s$ has neither a cycle nor a fixed point $\neq n-1$ and so is different from the transformations of Cases~2, 3 and~4. In Case~5, transformation $s^2$ maps a colliding pair to $n-1$, and here $s^2$ maps the unique colliding pair to $q_2 \neq n-1$. Thus, $s$ is different from the transformations of Case~5.

From~(a) we have the unique colliding pair $\{p,r\}$ focused to $q_1$. Then $q_1 < q_1 s = q_2$ means that $p < r$, and so $p$ is distinguished from $r$. Similarly, $q_1 > q_2$ means that $p > r$. Thus $0t=p$, $pt=n-1$, $q_1 t = r$, $q_2 t = r$, and $r t = n-1$, and $t$ is again uniquely defined from $s$.

\noindent{\hglue 15pt}
{\bf Case 7:} $t$ does not fit in any of the previous cases, but there are two states $q_1,q_2 \in Q \setminus \{0,n-1\}$ that are not fixed points and satisfy $q_1 t \neq n-1$ and $q_2 t \neq n-1$.\\
Since this is not Case~5 we may assume that $q_1t^2 = n-1$ and $q_2t^2 = n-1$.

Let $r_1 = q_1 t$ and $r_2 = q_2 t$; clearly $p \neq r_1$ and $p \neq r_2$. 
The in-degree of both $q_1$ and $q_2$ is $0$; otherwise $t$ would fit in Case~5.

We consider the following two sub-cases. If $p < r_1$ then {\bf(i)} 
let $s$ be the transformation illustrated in Fig.~\ref{fig:case7} and defined by
\begin{center}
$0 s = n-1, \quad p s = q_1, \quad r_1 s = q_1, \quad q_1 s = n-1,$

$q s = q t \text{ for the other states } q \in Q.$
\end{center}
If $p > r_1$ then {\bf(ii)} 
let $s$ be the transformation also illustrated in Fig.~\ref{fig:case7} and defined by
\begin{center}
$0 s = n-1, \quad p s = q_1, \quad r_1 s = q_1, \quad q_1 s = q_2,$

$q s = q t \text{ for the other states } q \in Q.$
\end{center}
\begin{figure}[ht]
\unitlength 10pt\footnotesize
\begin{center}\begin{picture}(28,22)(0,-1)
\gasset{Nh=2.5,Nw=2.5,Nmr=1.25,ELdist=0.5,loopdiam=2}
\node[Nframe=n](name)(0,20){$t\colon$}
\node(0)(2,13){0}
\node(p)(14,13){$p$}
\node(n-1)(26,13){$n-1$}
\node(q1)(10,16){$q_1$}
\node(q2)(10,20){$q_2$}
\node(r1)(18,16){$r_1$}
\node(r2)(18,20){$r_2$}
\drawedge(0,p){$t$}
\drawedge(p,n-1){$t$}
\drawloop[loopangle=90](n-1){$t$}
\drawedge(q1,r1){$t$}
\drawedge(q2,r2){$t$}
\drawedge(r1,n-1){$t$}
\drawedge(r2,n-1){$t$}

\node[Nframe=n](name)(0,9){$s\colon$}
\node(0')(2,2){0}
\node(p')(14,2){$p$}
\node(n-1')(26,2){$n-1$}
\node(q1')(10,5){$q_1$}
\node(q2')(10,9){$q_2$}
\node(r1')(18,5){$r_1$}
\node(r2')(18,9){$r_2$}
\drawedge[curvedepth=-2.5,dash={.5 .25}{.25},ELdist=-1](0',n-1'){$s$}
\drawloop[loopangle=90](n-1'){$s$}
\drawedge[dash={.5 .25}{.25}](p',q1'){$s$}
\drawedge[dash={.5 .25}{.25},ELdist=-1](r1',q1'){$s$}
\drawedge(q2',r2'){$s$}
\drawedge(r2',n-1'){$s$}
\drawedge[dash={.1 .25}{.25},ELdist=-1.5,curvedepth=-0.5](q1',n-1'){$s$(i)}
\drawedge[dash={.1 .25}{.25}](q1',q2'){$s$(ii)}
\end{picture}\end{center}
\caption{Case~7 in the proof of Theorem~\ref{thm:suffix-free_upper-bound}.}
\label{fig:case7}
\end{figure}
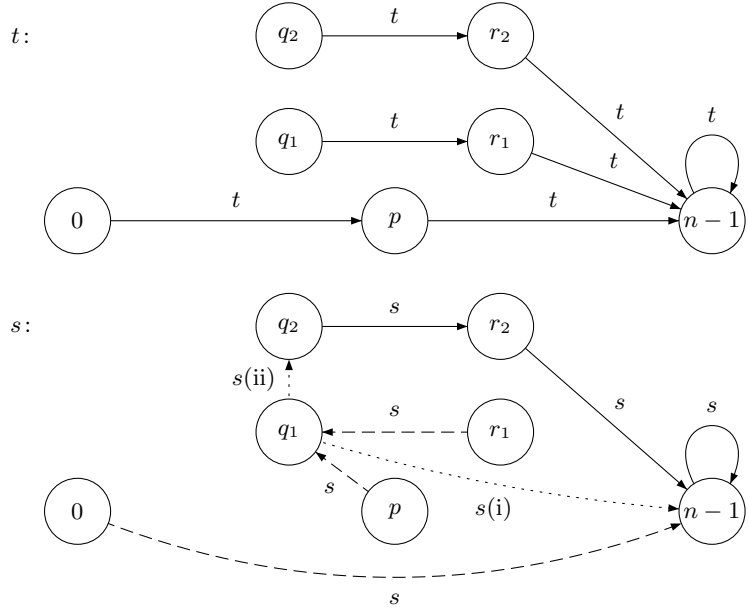

Observe that $s$ has the following properties:
\bd
\item[(a)] There is only one focused colliding pair, namely the pair $\{p,r_1\}$ mapped to $q_1$. Both states from the pair have in-degree $0$ in $s$.

Clearly $p$ and $r_1$ collide. 
No state is mapped by $s$ to $q_2$, and $0$ and $q_1$ are mapped to $n-1$ in (i). In (ii) $q_1$ is the only state mapped to $q_2$ and 0 is mapped to $n-1$. The other states are mapped exactly as in $t$. It follows that $s$ does not focus any other colliding pairs.

\item[(b)] The smallest $i$ with $ps^i = n-1$ is $2$ in (i), and is $4$ in (ii).

\item[(c)] $s$ has neither a cycle nor a fixed point with in-degree $\ge 2$ other than $n-1$.
\ed

Since $s$ focuses the colliding pair $\{p,r\}$, it is different from the transformations of Case~1. Also $s$ has neither a cycle nor a fixed point with in-degree $\ge 2$ other than $n-1$, and so is different from the transformations of Cases~2, 3 and~4. 
Here the states from the colliding pair have in-degree $0$, in contrast to the transformations of Case~5 (Property~(a) of Case~5). Now, observe that the smallest $i$ with $ps^i=n-1$ is $2$ or $4$, while in Case~6 it is $3$ (Property~(b) of Case~6).

From~(a) we have the unique colliding pair $\{p,r_1\}$ focused to $q_1$. If $ps^2 = n-1$, then $p < r_1$ (i), and so $p$ is distinguished from $r_1$. If $ps^2 \neq n-1$ then $ps^2 = q_2$, and $p > r_1$ (ii). Thus $0t=p$, $pt = n-1$, $r_1 t=n-1$, and $q_1 t = r_1$. Thus $t$ is uniquely defined from $s$.

\noindent{\hglue 15pt}
{\bf Case 8:} $t$ does not fit in any of the previous cases, but it has two fixed points $r_1$ and $r_2$ in $Q \setminus \{0,n-1\}$ with in-degree 1; assume that $r_1 < r_2$.\\

Let $s$ be the transformation illustrated in Fig.~\ref{fig:case8} and defined by
\begin{center}
$0 s = n-1, \quad p s = r_2, \quad r_1 s = r_2, \quad r_2 s = r_1,$

$q s = q t \text{ for the other states } q \in Q.$
\end{center}
\begin{figure}[ht]
\unitlength 10pt\footnotesize
\begin{center}\begin{picture}(28,19)(0,-1)
\gasset{Nh=2.5,Nw=2.5,Nmr=1.25,ELdist=0.5,loopdiam=2}
\node[Nframe=n](name)(0,16){$t\colon$}
\node(0)(2,11){0}
\node(p)(14,11){$p$}
\node(n-1)(26,11){$n-1$}
\node(r1)(17,14){$r_1$}
\node(r2)(22,14){$r_2$}
\drawedge(0,p){$t$}
\drawedge(p,n-1){$t$}
\drawloop[loopangle=90](r1){$t$}
\drawloop[loopangle=90](r2){$t$}

\node[Nframe=n](name)(0,7){$s\colon$}
\node(0')(2,2){0}
\node(p')(14,2){$p$}
\node(n-1')(26,2){$n-1$}
\node(r1')(17,5){$r_1$}
\node(r2')(22,5){$r_2$}
\drawedge[curvedepth=-2.5,dash={.5 .25}{.25},ELdist=-1](0',n-1'){$s$}
\drawedge[dash={.5 .25}{.25},ELdist=-1,curvedepth=-1.5](p',r2'){$s$}
\drawloop[loopangle=90](n-1'){$s$}
\drawedge[dash={.5 .25}{.25},curvedepth=1](r1',r2'){$s$}
\drawedge[dash={.5 .25}{.25},curvedepth=1](r2',r1'){$s$}
\end{picture}\end{center}
\caption{Case~8 in the proof of Theorem~\ref{thm:suffix-free_upper-bound}.}
\label{fig:case8}
\end{figure}
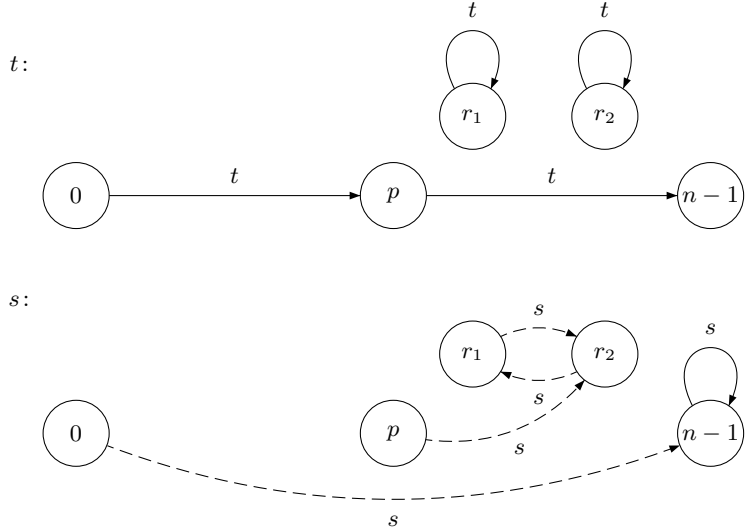

Observe that $s$ has the following properties:
\bd
\item[(a)] 
$\{p,r_1\}$ is the only focused colliding pair. The state to which the pair is focused lies on a cycle of length 2 and is not the minimal state in the cycle.

This follows from the fact that $r_1$, $r_2$, and $p$ are the only states in their orbit, only $p$ and $r_1$ are mapped to a single state, and $s$ differs from $t$ only by the mapping of 0 and of these three states.

\item[(b)] $s$ has a unique 2-cycle.
\ed

From~(a), $s \not\in T(n)$ and so $s$ is different from the transformations of Case~1. The transformations of Case~2 contain a cycle, but (Property~(a) of Case~2) the only colliding pair focused by them to a state lying on a cycle is mapped to the minimal state in cycles, in contrast to $s$ from this case. 
Also, $s$ has a cycle, and so differs from the transformations of Cases~3--7.

Again, $s$ uniquely defines $t$:
The orbit with the focused colliding pair contains precisely $p$, $r_1$, and $r_2$, and they are distinguished since $r_1<r_2$, and $r_1$ and $r_2$ lie on a cycle whereas $p$ does not.

\noindent{\hglue 15pt}
{\bf Case 9:} $t$ does not fit in any of the previous cases, but there is a state $q \in Q \setminus \{0,n-1\}$ that is not a fixed point and satisfies $qt \neq n-1$, $p < qt$, and there is a fixed point $f \neq n-1$.

Let $r = qt$; then $rt = n-1$ because otherwise this would fit in Case~5. Here $q$ is the only state from $Q \setminus \{0\}$ that is not a fixed point and  is not mapped to $n-1$, as otherwise $t$ would fit in Case~7. Similarly, $f$ is the only fixed point $\neq n-1$, as otherwise $t$ would fit in either Case~4 or Case~8.

Let $s$ be the transformation illustrated in Fig.~\ref{fig:case9} and defined by
\begin{center}
$0 s = n-1, \quad p s = r, \quad r s = q, \quad q s = p, \quad f s = r,$

$q s = q t \text{ for the other states } q \in Q.$
\end{center}
\begin{figure}[ht]
\unitlength 10pt\footnotesize
\begin{center}\begin{picture}(28,19)(0,-1)
\gasset{Nh=2.5,Nw=2.5,Nmr=1.25,ELdist=0.5,loopdiam=2}
\node[Nframe=n](name)(0,16){$t\colon$}
\node(0)(2,11){0}
\node(p)(14,11){$p$}
\node(n-1)(26,11){$n-1$}
\node(q)(11,14){$q$}
\node(r)(17,14){$r$}
\node(f)(22,14){$f$}
\drawedge(0,p){$t$}
\drawedge(p,n-1){$t$}
\drawedge(q,r){$t$}
\drawedge[ELdist=-1](r,n-1){$t$}
\drawloop[loopangle=90](f){$t$}

\node[Nframe=n](name)(0,7){$s\colon$}
\node(0')(2,2){0}
\node(p')(14,2){$p$}
\node(n-1')(26,2){$n-1$}
\node(q')(11,5){$q$}
\node(r')(17,5){$r$}
\node(f')(22,5){$f$}
\drawedge[curvedepth=-2.5,dash={.5 .25}{.25},ELdist=-1](0',n-1'){$s$}
\drawloop[loopangle=90](n-1'){$s$}
\drawedge[dash={.5 .25}{.25},curvedepth=-1,ELdist=-1](p',r'){$s$}
\drawedge[dash={.5 .25}{.25},curvedepth=-1,ELdist=-1](r',q'){$s$}
\drawedge[dash={.5 .25}{.25},curvedepth=-1,ELdist=-1](q',p'){$s$}
\drawedge[dash={.5 .25}{.25}](f',r'){$s$}
\end{picture}\end{center}
\caption{Case~9 in the proof of Theorem~\ref{thm:suffix-free_upper-bound}.}
\label{fig:case9}
\end{figure}
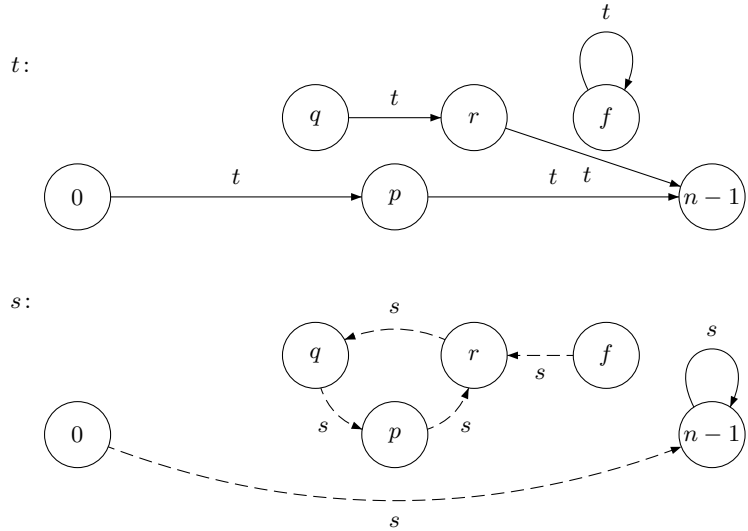

Observe that $s$ has the following properties:
\bd
\item[(a)] $\{p,f\}$ is the only focused colliding pair, and the state to which it is focused lies on a cycle of length 3 and is not the minimal state in the cycle.

\item[(b)] $s$ has a unique 3-cycle.
\ed

From~(a), $s \not\in T(n)$ and so $s$ is different from the transformations of Case~1. The transformations of Case~2 contain a cycle, but (Property~(a) of Case~2) the only colliding pair focused by them to a state lying on a cycle is focused to the minimal state appearing in a cycle, in contrast to $s$ from this case. 
Since $s$ has a cycle, it is different from the transformations of Cases~3--7. 
Also, $s$ has a unique 3-cycle in contrast with the transformations of Case~8, which have a unique 2-cycle (Property~(b) of Case~8).

Again, $s$ uniquely defines $t$: The orbit with the focused colliding pair contains precisely $p$, $q$, $r$ and $f$, and they are uniquely determined since $f$ is not in the 3-cycle and is mapped to $r$.

\noindent{\hglue 15pt}
{\bf Case 10:} $t$ does not fit in any of the previous cases, but there is a state $q \in Q \setminus \{0,n-1\}$ that is not a fixed point and satisfies $qt \neq n-1$, and a fixed point $f \in Q \setminus \{0,n-1\}$.\\
Let $r = qt$; then $rt = n-1$ since this is not Case~5. Now, in contrast to the previous case, we have $p > r$.

Let $s$ be the transformation illustrated in Fig.~\ref{fig:case10} and defined by
\begin{center}
$0 s = n-1, \quad p s = q, \quad r s = q, \quad q s = n-1,$

$q s = q t \text{ for the other states } q \in Q.$
\end{center}
\begin{figure}[ht]
\unitlength 10pt\footnotesize
\begin{center}\begin{picture}(28,19)(0,-1)
\gasset{Nh=2.5,Nw=2.5,Nmr=1.25,ELdist=0.5,loopdiam=2}
\node[Nframe=n](name)(0,16){$t\colon$}
\node(0)(2,11){0}
\node(p)(14,11){$p$}
\node(n-1)(26,11){$n-1$}
\node(q)(11,14){$q$}
\node(r)(17,14){$r$}
\node(f)(22,14){$f$}
\drawedge(0,p){$t$}
\drawedge(p,n-1){$t$}
\drawedge(q,r){$t$}
\drawedge[ELdist=-1](r,n-1){$t$}
\drawloop[loopangle=90](f){$t$}

\node[Nframe=n](name)(0,7){$s\colon$}
\node(0')(2,2){0}
\node(p')(14,2){$p$}
\node(n-1')(26,2){$n-1$}
\node(q')(11,5){$q$}
\node(r')(17,5){$r$}
\node(f')(22,5){$f$}
\drawedge[curvedepth=-2.5,dash={.5 .25}{.25},ELdist=-1](0',n-1'){$s$}
\drawloop[loopangle=90](n-1'){$s$}
\drawedge[dash={.5 .25}{.25},curvedepth=0](p',q'){$s$}
\drawedge[dash={.5 .25}{.25},curvedepth=0,ELdist=-1](r',q'){$s$}
\drawedge[dash={.5 .25}{.25},curvedepth=-0.5,ELdist=-1](q',n-1'){$s$}
\drawloop[loopangle=90](f'){$s$}
\end{picture}\end{center}
\caption{Case~10 in the proof of Theorem~\ref{thm:suffix-free_upper-bound}.}
\label{fig:case10}
\end{figure}
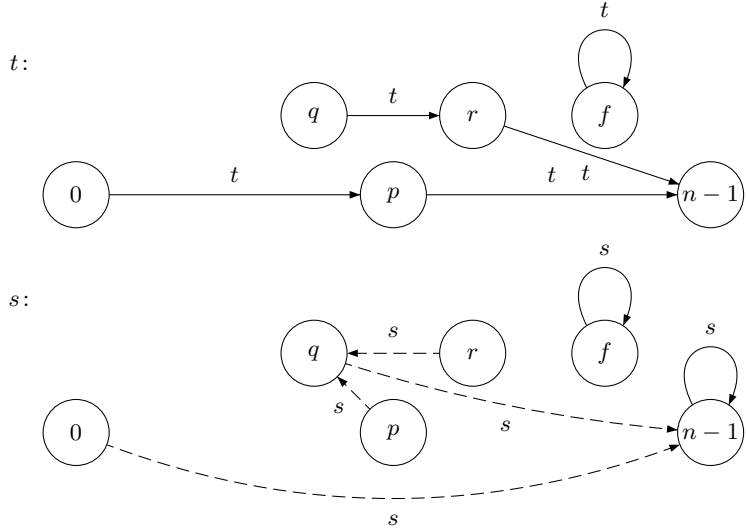

Observe that $s$ has the following properties:
\bd
\item[(a)] $\{p,r\}$ is the only focused colliding pair.

\item[(b)] $s$ does not have cycles, and each state $\in Q \setminus \{p,r,f\}$ is mapped to $n-1$.

\item[(c)] $s$ has the fixed point $f \neq n-1$.
\ed

From~(a), $s \not\in T(n)$ and so $s$ is different from the transformations of Case~1. Since the transformations of Cases~2, 8, and~9 contain cycles, they are different from $s$. Here the unique focused colliding pair is not mapped to a fixed point, in contrast with the transformations of Cases~3 and 4. Since both states from the pair have in-degree $0$, $s$ is different from the transformation of Case~5. For a distinction with Case~6, observe that the smallest $i$ such that $ps^i=n-1$ is $2$, in contrast with $3$ (Property~(b) of Case~6). For a distinction with Case~7, observe that besides the focused colliding states $p$ and $r$, there is no state $q'$ that is not a fixed point and satisfies $q's \neq n-1$, whereas in the transformation of Case~7 $q_2$ is such a state.

Again, $s$ uniquely defines $t$: Here $\{p,r\}$ is the only focused colliding pair, and $p$ is distinguished as the larger state.

\noindent{\hglue 15pt}
{\bf Case 11:} $t$ does not fit in any of the previous cases, but there is a state $q \in Q \setminus \{0,n-1\}$ that is not a fixed point and satisfies $qt \neq n-1$.\\
As shown in Case~9, $q$ is the only state from $Q \setminus \{0\}$ that is not mapped to $n-1$, and also $t$ has no fixed points other than $n-1$, as otherwise it would fit in one of the previous cases. Hence, all states from $Q \setminus \{0,q\}$ are mapped to $n-1$.
Let $r = qt$.

Here we use the assumption that $n \ge 7$. So in $Q \setminus \{0,p,q,r,n-1\}$ we have at least 2 states, say $r_1$ and $r_2$, that are mapped to $n-1$.

We consider the following two sub-cases:\\
{\bf Sub-case (i):} $p < r$.\\
Let $s$ be the transformation illustrated in Fig.~\ref{fig:case11} and defined by
\begin{center}
$0 s = n-1, \quad p s = q, \quad r s = q, \quad q s = n-1,$

$q s = q t \text{ for the other states } q \in Q.$
\end{center}
{\bf Sub-case (ii):} $p > r$.\\
Let $s$ be the transformation also illustrated in Fig.~\ref{fig:case11} and defined by
\begin{center}
$0 s = n-1, \quad p s = q, \quad r s = q, \quad q s = n-1, \quad r_1 s = r_2, \quad r_2 s = r_1,$

$q s = q t \text{ for the other states } q \in Q.$
\end{center}

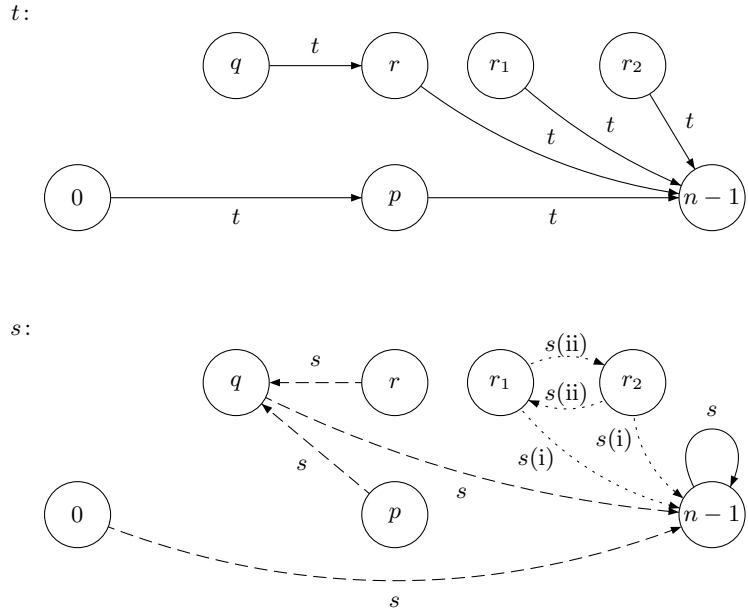
\begin{figure}[ht]
\unitlength 10pt\footnotesize
\begin{center}\begin{picture}(28,23)(0,-1)
\gasset{Nh=2.5,Nw=2.5,Nmr=1.25,ELdist=0.5,loopdiam=2}
\node[Nframe=n](name)(0,21){$t\colon$}
\node(0)(2,14){0}
\node(p)(14,14){$p$}
\node(n-1)(26,14){$n-1$}
\node(q)(8,19){$q$}
\node(r)(14,19){$r$}
\node(r1)(18,19){$r_1$}
\node(r2)(23,19){$r_2$}
\drawedge[ELdist=-1](0,p){$t$}
\drawedge[ELdist=-1](p,n-1){$t$}
\drawedge(q,r){$t$}
\drawedge[curvedepth=-1](r,n-1){$t$}
\drawedge[curvedepth=-.5](r1,n-1){$t$}
\drawedge[curvedepth=0](r2,n-1){$t$}

\node[Nframe=n](name)(0,9){$s\colon$}
\node(0')(2,2){0}
\node(p')(14,2){$p$}
\node(n-1')(26,2){$n-1$}
\node(q')(8,7){$q$}
\node(r')(14,7){$r$}
\node(r1')(18,7){$r_1$}
\node(r2')(23,7){$r_2$}
\drawedge[curvedepth=-2.5,dash={.5 .25}{.25},ELdist=-1](0',n-1'){$s$}
\drawloop[loopangle=90](n-1'){$s$}
\drawedge[dash={.5 .25}{.25},curvedepth=0](p',q'){$s$}
\drawedge[dash={.5 .25}{.25},curvedepth=0,ELdist=-1](r',q'){$s$}
\drawedge[dash={.5 .25}{.25},curvedepth=-1,ELdist=-1](q',n-1'){$s$}
\drawedge[dash={.1 .25}{.25},curvedepth=-1,ELdist=-1.7,ELpos=30](r1',n-1'){$s$(i)}
\drawedge[dash={.1 .25}{.25},curvedepth=-1,ELdist=-1.7,ELpos=30](r2',n-1'){$s$(i)}
\drawedge[dash={.1 .25}{.25},curvedepth=1,ELdist=0](r1',r2'){$s$(ii)}
\drawedge[dash={.1 .25}{.25},curvedepth=1,ELdist=-1](r2',r1'){$s$(ii)}
\end{picture}\end{center}
\caption{Case~11 in the proof of Theorem~\ref{thm:suffix-free_upper-bound}.}
\label{fig:case11}
\end{figure}
Observe that $s$ has the following properties:
\bd
\item[(a)] $\{p,r\}$ is the only  focused  colliding pair  and it is focused to state  $q$ with $qs = n-1$.

\item[(b)] $s$ does not have any fixed points other than $n-1$, and in (i) $s$ does not have any cycles, whereas in (ii) $s$ has a cycle but no colliding pair from the orbit of the cycle is  focused.
\ed

From~(a), $s \not\in T(n)$ and so $s$ is different from the transformations of Case~1. In (ii) $s$ has a cycle but the colliding pair is not from the orbit of the cycle, in contrast to Case~2. The transformations of Cases~3--10 do not have a cycle whose orbit has no  focused colliding pairs.

In (i) $s$ does not have a cycle, and so is different from the transformations of Cases~2, 8 and 9. Also, there is no fixed point other than $n-1$, so $s$ is different from the transformations of Cases~3 and 4. Since both $p$ and $r$ from the colliding pair have in-degree $0$, $s$ is different from the transformations of Case~5. Since the smallest $i$ such that $ps^i = n-1$ is $2$, $s$ is different from the transformations of Case~6. For a distinction with Case~7, observe that besides the colliding states $p$ and $r$, there is no state $q'$ that is not a fixed point and satisfies $q's \neq n-1$, whereas in the transformation of Case~7 $q_2$ is such a state.

The transformations of Cases~2, 8, and~9 contain cycles whose orbits contain a focused colliding pair. Here $s$ in (i) does not have a cycle, and in (ii) has a 2-cycle but the unique colliding pair is not in its orbit. Also, the unique  focused colliding pair  is not mapped to a fixed point, in contrast with the transformations of Cases~3, 4 and 10. Since both states from the pair have in-degree $0$, $s$ is different from the transformation of Case~5. For a distinction with Case~6, observe that the smallest $i$ such that $ps^i=n-1$ is $2$, in contrast with $3$ (Property~(b) of Case~6). For a distinction with Case~7, observe that besides the colliding states $p$ and $r$, there is no state $q'$ that is not a fixed point and satisfies $q's \neq n-1$, whereas in the transformation of Case~7 $q_2$ is such a state.

Again, $s$ uniquely defines $t$: Here $\{p,r\}$ is the only  focused colliding pair, and if we have a 2-cycle, then $p$ is distinguished as the smaller state; otherwise $p$ is the larger one from the pair.

\noindent{\hglue 15pt}
{\bf Case 12:} $t$ does not fit in any of the previous cases.\\
Here $t$ must contain exactly one fixed point $f \in Q \setminus \{n-1\}$, and every state from $Q \setminus \{0,f\}$ is mapped to $n-1$. 
If all states from $Q \setminus \{0\}$ would be mapped to $n-1$, then by Proposition~\ref{prop:sf_wit}, $t$ would be in $\Wsf(n)$ and so would fit in Case~1.

Because $n \ge 7$, in $Q \setminus \{0,p,f,n-1\}$ we have at least 2 states, say $r_1$ and $r_2$, that are mapped to $n-1$.

Let $s$ be the transformation illustrated in Fig.~\ref{fig:case12} and defined by
\begin{center}
$0 s = n-1, \quad p s = f, \quad r_1 s = r_2, \quad r_2 s = r_1,$

$q s = q t \text{ for the other states } q \in Q.$
\end{center}
\begin{figure}[ht]
\unitlength 10pt\footnotesize
\begin{center}\begin{picture}(28,22)(0,-1)
\gasset{Nh=2.5,Nw=2.5,Nmr=1.25,ELdist=0.5,loopdiam=2}
\node[Nframe=n](name)(0,19){$t\colon$}
\node(0)(2,13){0}
\node(p)(14,13){$p$}
\node(n-1)(26,13){$n-1$}
\node(f)(8,17){$f$}
\node(r1)(14,17){$r_1$}
\node(r2)(18,17){$r_2$}
\node[Nframe=n](rdots)(22,17){$\dots$}
\drawedge[ELdist=-1](0,p){$t$}
\drawedge[ELdist=-1](p,n-1){$t$}
\drawloop[loopangle=90](f){$t$}
\drawedge[curvedepth=-1](r1,n-1){$t$}
\drawedge[curvedepth=-0.5](r2,n-1){$t$}
\drawedge(rdots,n-1){$t$}

\node[Nframe=n](name)(0,8){$s\colon$}
\node(0')(2,2){0}
\node(p')(14,2){$p$}
\node(n-1')(26,2){$n-1$}
\node(f')(8,6){$f$}
\node(r1')(14,6){$r_1$}
\node(r2')(18,6){$r_2$}
\node[Nframe=n](rdots')(22,6){$\dots$}
\drawedge[curvedepth=-2.5,dash={.5 .25}{.25},ELdist=-1](0',n-1'){$s$}
\drawloop[loopangle=90](n-1'){$s$}
\drawloop[loopangle=90](f'){$s$}
\drawedge[dash={.5 .25}{.25},curvedepth=0](p',f'){$s$}
\drawedge[dash={.5 .25}{.25},curvedepth=1](r1',r2'){$s$}
\drawedge[dash={.5 .25}{.25},curvedepth=1](r2',r1'){$s$}
\drawedge(rdots',n-1'){$s$}
\end{picture}\end{center}
\caption{Case~12 in the proof of Theorem~\ref{thm:suffix-free_upper-bound}.}
\label{fig:case12}
\end{figure}
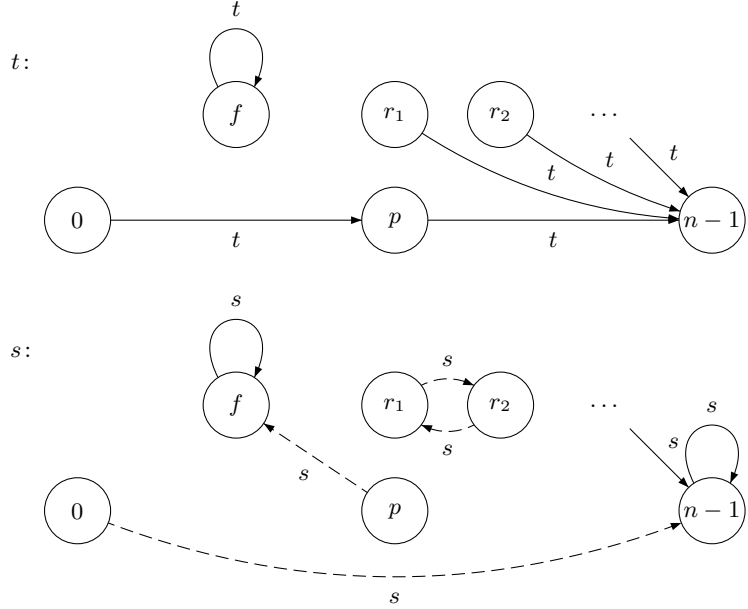

Observe that $s$ has the following properties:
\bd
\item[(a)] $\{p,f\}$ is the only  focused colliding pair, focused to  the fixed point $f$.

\item[(b)] $s$ has a 2-cycle, but no colliding pair from the orbit of the cycle is  focused.
\ed

From~(a), $s \not\in T(n)$ and so $s$ is different from the transformations of Case~1. Here $s$ has a cycle but the colliding pair is not from the orbit of the cycle, in contrast to Case~2. The transformations of Cases~3--10 do not have a cycle whose orbit has no focused colliding pairs. The transformations of Case~11 have such a cycle, but the orbit with the focused colliding pair is the orbit of fixed point $n-1$, and here it is the orbit of $f \neq n-1$.

Again, $s$ uniquely defines $t$: Here $\{p,f\}$ is the only  focused colliding pair, and $p$ is distinguished from $f$ as it is not a fixed point. Hence we can define $0t = p$, $pt = n-1$, $r_1 t = n-1$, and $r_2 t = n-1$.
\end{proof}

\section{Uniqueness of maximal witness}

Our third contribution is a proof that the transition semigroup of a DFA $\cD_n=(Q, \Sigma, \delta, 0,F)$ of a suffix-free language with syntactic complexity $(n-1)^{n-2}+n-2$ is unique.

\begin{lemma}
\label{lem:colliding0}
If $n\ge 4$ and $\cD_n$ has no colliding pairs, then $|T(n)|\le (n-1)^{n-2}+n-2$ and $T(n)$ is a subsemigroup of $\Wsf(n)$.
\end{lemma}
\begin{proof}
Consider an arbitrary transformation $t\in T(n)$ and let $p=0t$. 
If $p=n-1$, then any state other than $0$ and $n-1$ can possibly be mapped by $t$ to any one of the $n-1$ states other than 0 (0 is not possible in view of Lemma~\ref{lem:sf}); hence there are at most $(n-1)^{n-2}$ such transformations. All of these transformations are in $\Wsf(n)$ by 
Proposition~\ref{prop:sf_wit}.

If $p\neq n-1$, then $qt = n-1$ for any $q \neq 0$, because there are no colliding pairs.
Thus $t = (Q\setminus \{0\} \to n-1)(0\to p)$,
and all of $n-2$ such transformations are in $\Wsf(n)$ by Proposition~\ref{prop:sf_wit}.
It follows that $T(n)$ is a subsemigroup of $\Wsf(n)$ and has size at most $(n-1)^{n-2}+n-2$.
\end{proof}

\begin{lemma}
If $n \le 7$ and $\cD_n$ has at least one colliding pair, then $|T(n)| < (n-1)^{n-2} + n - 2$.
\end{lemma}
\begin{proof}
Let $\varphi$ be the injective function from the proof of Theorem~\ref{thm:suffix-free_upper-bound} and 
assume that there is a colliding pair $\{p,r\}$.
Let $r_1$, $r_2$ and $r_3$ be three distinct states from $Q \setminus \{0,p,r,n-1\}$; there are at least 3 such states since $n \ge 7$.
Let $s$ be the following transformation:
\begin{center}
$0 s = n-1, \quad p s = r, \quad  rs=r,   \quad r_1 s = r_2, \quad r_2 s = r_3, \quad r_3 s = r_1,$

$q s = q t \text{ for the other states } q \in Q.$
\end{center}

We can show that $s$ is not defined in any case in the proof of Theorem~\ref{thm:suffix-free_upper-bound}.
Note that $s$ focuses the colliding pair $\{p,r\}$, and so it cannot be present in $T(n)$; hence it is not defined in Case~1. We can follow the proof of injectivity of the transformations in Case~12 of Theorem~\ref{thm:suffix-free_upper-bound}, and show that $s$ is different from all the transformations of Cases~2--11. For a distinction from the transformations of Case~12, observe that they each  have a  2-cycle, and here $s$ has a  3-cycle.

Thus $s \not\in \varphi(T(n))$, but $s \in \Wsf(n)$, and so $\varphi(T(n)) \subsetneq \Wsf(n)$. Since $\varphi$ is injective, it follows that $|T(n)| < |\Wsf(n)| = (n-1)^{n-2}+n-2$.
\end{proof}

\begin{corollary}
\label{cor:unique}
For $n\ge 7$, the maximal transition semigroups of DFAs of suffix-free languages are unique.
\end{corollary}

Finally, we show that $\Sig$ cannot have fewer than five letters.

\begin{theorem}
If $n\ge 7$,  $\cD_n=(Q, \Sigma, \delta, 0,F)$ is a minimal DFA of a suffix-free language, and $|\Sig|<5$, then $|T(n)|<(n-1)^{n-2}+n-2$.
\end{theorem}
\begin{proof}
DFA $\cD_n$ has the initial state 0, and an empty state, say $n-1$.
Let $M$ be the set of the remaining $n-2$ ``middle'' states.
From Lemma~\ref{lem:sf} no transformation can map any state in $Q$ to 0, and every transformation fixes $n-1$. 

Suppose the upper bound $(n-1)^{n-2}+n-2$ is reached by $T(n)$. 
From Proposition~\ref{prop:sf_wit} and Corollary~\ref{cor:unique} all transformations of $M$ must be possible, and it is well known that three generators are necessary to achieve this. Let the letters $a$, $b$, and $c$ correspond to these three generators, $t_a$, $t_b$ and $t_c$.
If $0 t_a \neq n-1$, then $t_a$ must be a transformation of type (b) from Proposition~\ref{prop:sf_wit}, and so $q t_a = n-1$ for any $q \in M$. So $t_a$ cannot be a generator of a transformation of $M$.
Hence we must have $0t_a=n-1$, and also $0t_b=0t_c=n-1$.

So far, the states in $M$ are not reachable from 0; hence there must be a letter, say $e$, such that $0t_e=p$ is in $M$.
This must be a transformation of type (b) from Proposition~\ref{prop:sf_wit}, and all the states of $M$ must be mapped to $n-1$ by $t_e$.

Finally, to reach the upper bound we must be able to map any proper subset of $M$ to $n-1$.
The letter $e$ will not do, since it maps \emph{all} states of $M$ to $n-1$. Hence we require a fifth letter, say~$d$.
\end{proof}

\subsection{Uniqueness of small semigroups}

Here we consider maximal transition semigroups of DFAs  having six or fewer states and accepting suffix-free languages.
Recall that every transformation in a transition semigroup of a DFA with $n$ states accepting a suffix-free language must be in $\mathrm{\mathbf{B}_{sf}}(n)$.

There is only one transformation in $\mathrm{\mathbf{B}_{sf}}(2)$, and  three in $\mathrm{\mathbf{B}_{sf}}(3)$.
Since $\mathrm{\mathbf{B}_{sf}}(2)$ and $\mathrm{\mathbf{B}_{sf}}(3)$ are semigroups, the maximal semigroups for $n=2$ and $n=3$ have 1 and 3 elements, respectively, and are unique.
From~\cite{BLY12} it is known that $\Vsf(n) = \Wsf(n)= \mathrm{\mathbf{B}_{sf}}(n)$ if $n\in \{2,3\}$.

From~\cite{BLY12} it is also known that 
$\Vsf(n)$ is largest  for $n \in \{4, 5\}$ and that $\Wsf(n)$ is largest for $n = 6$.
We now prove that  $\Vsf(4)$, $\Vsf(5)$, and $\Wsf(6)$ are the unique largest semigroups for those values of $n$.

We say that  transformations $t$ and $t'$ in $\mathrm{\mathbf{B}_{sf}}(n)$ \emph{conflict} if,
whenever $t$ and $t'$ belong to the transition semigroup of a DFA $\cD$, then one of the following conditions holds:
\begin{enumerate}
\item The language accepted by $\cD$ is not suffix-free.
\item Every two states from $\{1,\ldots,n-2\}$ are colliding.
\item Every two states from $\{1,\ldots,n-2\}$ are focused.
\end{enumerate}

\begin{lemma}\label{lem:largest_other}
In a largest transition semigroup $X_n$ of a DFA $\cD_n=(Q, \Sigma, \delta, 0,F)$ of a suffix-free language there are no conflicting pairs of transformations, unless $X_n = \Vsf$ or $X_n = \Wsf$.
\end{lemma}
\begin{proof}
Obviously, there are no conflicting pair of transformations because of (1).
If a pair of transformations conflicts because of (2), then from Proposition~\ref{prop:unique_all_colliding} a largest transition semigroup must be $\Vsf$.
If a pair of transformations conflicts because of (3), then no pair of states is colliding, since it is focused; from Proposition~\ref{prop:unique_no_colliding} a largest transition semigroup must be $\Wsf$.
\end{proof}

A transformation $t$ of $Q$ is called \emph{semiconstant} if it maps a non-empty subset $S \subseteq Q$ to a single state $q \in Q$, and fixes $Q \setminus S$. We denote it by $(S \to q)$.

\begin{lemma}\label{lem:semiconstant_in_suffix-free}
In the transition semigroup of a DFA $\cD_n=(Q, \Sigma, \delta, 0,F)$ of a suffix-free language $L$ all semiconstant transformations $t=(S \to q)$ are such that $0 \in S$ and $q = n-1$. Moreover, every such transformation is present if the transition semigroup is maximal.
\end{lemma}
\begin{proof}
If $0 \not\in S$, then $0t=0$; this implies that $t$ is not in $\mathrm{\mathbf{B}_{sf}}(n)$, contradicting our assumption that  $L$ is suffix-free.

If $0 \in S$ and $q \neq n-1$, then $0t=q$ and $q$ is a fixed point in $\{1,\ldots,n-2\}$.
Since $q$ is non-empty, it accepts some word $x$; then $tx$ and $ttx$ are both in $L$ contradicting that $L$ is suffix-free.

Thus we must have  $0 \in S$ and $q = n-1$.
We now argue that every such transformation $t$ must be present in a maximal semigroup by showing that $t$ can always be added if not present.
Suppose that the addition of a letter $a$ that induces $t$ results in a DFA that does not accept a suffix-free language.
Then there must be two words $v$ and $uv$ in the language, such that either $u$ or $v$ contains $a$; let $u$ and $v$ be such that $uv$ is a shortest such word.
Suppose $v=v_1av_2$; since $t$ maps some states to $n-1$ and fixes all the others and  $v$ is accepted, $a$ must map $\delta(0,v_1)$ to itself, and hence can be removed, leaving $v'=v_1v_2$ in $L$.
Similarly, if $u=u_1au_2$ we can remove $a$, obtaining $u'=u_1u_2$. Then $u'v'$ and $v'$ are in 
$L$, and $u'v'$ is shorter than $uv$---a contradiction.
\end{proof}

\begin{theorem}
For $n = 6$, the maximal transition semigroup of DFAs of suffix-free languages is $\Wsf(6)$ and it is unique.
For $n \le 5$, the maximal transition semigroup of DFAs of suffix-free languages is $\Vsf(n)$ and it is unique.
\end{theorem}
\begin{proof}
We have verified this with the help of computation.
We used the idea of conflicting pairs of transformations from~\cite[Theorem~20]{BLY12}, and 
we enumerated non-isomorphic DFAs using the approach 
of~\cite{KiSz2013}.
For $n=6$ there are $\mathrm{|\mathbf{B}_{sf}(6)|}=1169$ transformations that can be present in the transition semigroup of a DFA of a suffix-free language; so it is not possible to check all maximal subsets of $\mathrm{\mathbf{B}_{sf}}(6)$ in a naive way.
Our method is as follows.

Say that an automaton $\cD'(Q,\Sigma',\delta')$ is an \emph{extension} of a semiautomaton $\cD(Q,\Sigma,\delta)$ if it can be obtained from $\cD$ by adding letters, that is, if $\Sigma \subset \Sigma'$ and $\delta' \subset \delta$. The same concept can be extended to DFAs.
We start from the set $A_1$ of all non-isomorphic unary semiautomata with $n$-states.
Given a set $A_i$ of semiautomata over an $i$-ary alphabet, using~\cite{KiSz2013} we generate all of their non-isomorphic extensions $A_{i+1}$ over an $(i+1)$-ary alphabet.
To reduce the number of generated semiautomata and make the whole computation possible, we check for every semiautomaton whether there may exist an extension of it such that after adding initial and final states, the extension accepts a suffix-free language and its transition semigroup is a largest one but different from $\Wsf$ (if $n=6$) and $\Vsf$ (if $n \le 5$).
Also, we check whether the transition semigroup of the semiautomaton is irreducibly generated by the transitions of the letters, that is, all these transitions are required to generate the semigroup. Note that if a DFA $\cD$ has a transition semigroup whose generating set can be reduced, then there exists a DFA $\cD'$ over a smaller alphabet that has the same semigroup; moreover, $\cD'$ 
recognizes a suffix-free language if and only if $\cD$ does.

For every generated semiautomaton we test all selections of the initial state and the empty state, relabel them to $0$ and $n-1$, and check the resulting DFAs. If all the DFAs are rejected, then we skip the semiautomaton.
First, we check if the DFA accepts a suffix-free language (see~\cite{BSX10} for testing).
Then, we compute a rough bound on the maximal size of the transition semigroups that are different from $\Wsf$ and $\Vsf$ of extensions of the DFA.
If this is smaller than the syntactic complexity, we reject the DFA.
Extending the idea from~\cite[Theorem~20]{BLY12}, we compute the bound in the following way:
\begin{enumerate}
\item We compute the transition semigroup $X_n$ of the DFA.
\item For every transformation from $t \in B_{sf} \setminus X_n$ we check whether adding $t$ as a generator results in a DFA that accepts a suffix-free language, and neither all pairs of states are colliding nor all are focused.
Otherwise, by~Lemma~\ref{lem:largest_other} we  can omit it.
Let $Y_n$ be the set of allowed transformations.
Note that $|X_n|+|Y_n|$ gives a rough bound for the size of the maximal transition semigroups.
\item We compute a matching $M$ in the graph of conflicts of transformations induced by $Y_n$:
Two transformations $t,t' \in Y_n$ can be matched if they conflict.
We compute it by a simple greedy algorithm in $O(|Y_n|^2)$ time.
\item We finally use $|X_n|+|Y_n|-|M|$ as the bound.
\end{enumerate}

Finally, by Lemma~\ref{lem:semiconstant_in_suffix-free} we remove semiconstant transformations from the set $A_1$ of transformations  that is used to generate non-isomorphic semiautomata. Then we always add all allowed semicontant transformations to $X_n$ in step~(1).

Using this method, for $n=6$ we had to verify semiautomata only up to $9$ letters, and the computation took a few minutes.
\end{proof}

\section{Conclusions}

We have shown that the upper bound on the syntactic complexity of suffix-free languages is $(n-1)^{n-2}+n-2$ for $n \ge 6$. Since it was known that this is also a lower bound, our result settles the problem. Moreover, we have proved that an alphabet of at least five letters is necessary to reach the upper bound, and that the maximal transition semigroups are unique for every $n$.

\section*{Acknowledgements}
          
This work was supported by the Natural Sciences and Engineering Research Council of Canada (NSERC)
grant No.~OGP000087, and by Polish NCN grant 2014/15/B/ST6/00615.

\bibliographystyle{splncs03}

\begin{thebibliography}{10}
\providecommand{\url}[1]{\texttt{#1}}
\providecommand{\urlprefix}{URL }

\bibitem{AnBr09}
Ang, T., Brzozowski, J.: Languages convex with respect to binary relations, and
  their closure properties. Acta Cybernet.  19(2),  445--464 (2009)

\bibitem{BPR09}
Berstel, J., Perrin, D., Reutenauer, C.: Codes and Automata. Cambridge
  University Press (2009)

\bibitem{Brz10}
Brzozowski, J.: Quotient complexity of regular languages. J. Autom. Lang. Comb.
   15(1/2),  71--89 (2010)

\bibitem{BrSz15a}
Brzozowski, J., Szyku{\l}a, M.: Complexity of suffix-free regular languages.
  In: Kosowski, A., Walukiewicz, I. (eds.) FCT 2015. LNCS, vol. 9210, pp.
  146--159. Springer (2015), full version
  at~\url{http://arxiv.org/abs/1504.05159}

\bibitem{BrSz15}
Brzozowski, J., Szyku{\l}a, M.: Upper bound on syntactic complexity of
  suffix-free languages. In: Shallit, J. (ed.) DCFS 2015. LNCS, vol. 9118, pp.
  33--45. Springer (2015)

\bibitem{BrYe11}
Brzozowski, J., Ye, Y.: Syntactic complexity of ideal and closed languages. In:
  Mauri, G., Leporati, A. (eds.) DLT 2011. LNCS, vol. 6795, pp. 117--128.
  Springer (2011)

\bibitem{BLY12}
Brzozowski, J., Li, B., Ye, Y.: Syntactic complexity of \mbox{prefix-,}
  \mbox{suffix-,} \mbox{bifix-,} and factor-free regular langauges. Theoret.
  Comput. Sci.  449,  37--53 (2012)

\bibitem{BSX10}
Brzozowski, J., Shallit, J., Xu, Z.: Decision problems for convex languages.
  Information and Computation  209,  353--367 (2011)

\bibitem{GaMa09}
Ganyushkin, O., Mazorchuk, V.: Classical Finite Transformation Semigroups: An
  Introduction. Springer (2009)

\bibitem{HaSa09}
Han, Y.S., Salomaa, K.: State complexity of basic operations on suffix-free
  regular languages. Theoret. Comput. Sci.  410(27-29),  2537--2548 (2009)

\bibitem{KiSz2013}
Kisielewicz, A., Szyku{\l}a, M.: Generating small automata and the \v{C}ern\'y
  conjecture. In: Konstantinidis, S. (ed.) CIAA 2013. LNCS, vol. 7982, pp.
  340--348. Springer (2013)

\bibitem{Pin97}
Pin, J.E.: Syntactic semigroups. In: Handbook of Formal Languages, vol.~1:
  Word, Language, Grammar, pp. 679--746. Springer, New York, NY, USA (1997)

\bibitem{Thi73}
Thierrin, G.: Convex languages. In: Nivat, M. (ed.) Automata, Languages and
  Programming, pp. 481--492. North-Holland (1973)

\bibitem{Yu01}
Yu, S.: State complexity of regular languages. J. Autom. Lang. Comb.  6,
  221--234 (2001)

\end{thebibliography}
\providecommand{\noopsort}[1]{}

\end{document}